\documentclass[11pt,oneside,english]{amsart}
\usepackage[T1]{fontenc}
\usepackage[latin9]{inputenc}
\usepackage{geometry}
\usepackage{amsthm}
\usepackage{amstext}
\usepackage{amssymb}
\usepackage{graphicx}
\usepackage{esint}

\makeatletter
\numberwithin{equation}{section}
\numberwithin{figure}{section}
  \theoremstyle{plain}
  \newtheorem*{lem*}{\protect\lemmaname}
  \theoremstyle{plain}
  \newtheorem*{prop*}{\protect\propositionname}
  \theoremstyle{plain}
  \newtheorem*{thm*}{\protect\theoremname}
\theoremstyle{plain}
\newtheorem{thm}{\protect\theoremname}
  \theoremstyle{definition}
  \newtheorem{defn}[thm]{\protect\definitionname}
  \theoremstyle{plain}
  \newtheorem{lem}[thm]{\protect\lemmaname}
  \theoremstyle{plain}
  \newtheorem{prop}[thm]{\protect\propositionname}
  \theoremstyle{plain}
  \newtheorem{cor}[thm]{\protect\corollaryname}
  \theoremstyle{remark}
  \newtheorem*{rem*}{\protect\remarkname}

\makeatother

\usepackage{babel}
  \providecommand{\corollaryname}{Corollary}
  \providecommand{\definitionname}{Definition}
  \providecommand{\lemmaname}{Lemma}
  \providecommand{\propositionname}{Proposition}
  \providecommand{\remarkname}{Remark}
  \providecommand{\theoremname}{Theorem}
\providecommand{\theoremname}{Theorem}

\begin{document}
\global\long\def\bK{\mathbb{K}}
 \global\long\def\bC{\mathbb{C}}
 \global\long\def\bR{\mathbb{R}}
 \global\long\def\bZ{\mathbb{Z}}
 \global\long\def\bN{\mathbb{N}}
 \global\long\def\bQ{\mathbb{Q}}
 \global\long\def\bH{\mathbb{H}}

\global\long\def\half{\frac{1}{2}}
 \global\long\def\ii{\mathfrak{i}}

\global\long\def\bdry{\partial}
 \global\long\def\cl#1{\overline{#1}}

\global\long\def\PR{\mathsf{P}}
 \global\long\def\EX{\mathsf{E}}
 \global\long\def\sU{\mathcal{U}}

\global\long\def\GL{\mathrm{GL}}
 \global\long\def\SL{\mathrm{SL}}
 \global\long\def\gl{\mathfrak{gl}}
 \global\long\def\sl{\mathfrak{sl}}
 \global\long\def\TL{\mathrm{TL}}

\global\long\def\Hom{\mathrm{Hom}}
 \global\long\def\End{\mathrm{End}}
 \global\long\def\Aut{\mathrm{Aut}}
 \global\long\def\Rad{\mathrm{Rad}}
 \global\long\def\Ext{\mathrm{Ext}}

\global\long\def\Kern{\mathrm{Ker}}
 \global\long\def\Imag{\mathrm{Im}}

\global\long\def\Alt{\mathrm{Alt}}

\global\long\def\dmn{\mathrm{dim}}
 \global\long\def\spn{\mathrm{span}}
 \global\long\def\tens{\otimes}
 \global\long\def\Mat{\mathrm{Mat}}
 \global\long\def\unitmat{\mathbb{I}}
 \global\long\def\id{\mathrm{id}}
 \global\long\def\diag{\mathrm{diag}}
 \global\long\def\unit{\mathbf{1}}

\global\long\def\set#1{\left\{  #1\right\}  }
 \global\long\def\sgn{\mathrm{sgn} }

\global\long\def\re{\Re\mathfrak{e}}
 \global\long\def\im{\Im\mathfrak{m}}
 \global\long\def\ud{\mathrm{d}}
 \global\long\def\arg{\mathrm{arg}}
 \global\long\def\isom{\cong}

\global\long\def\op{\mathrm{op}}
 \global\long\def\cop{\mathrm{cop}}

\global\long\def\eps{\varepsilon}
 \global\long\def\const{\mathrm{const.}}

\global\long\def\binomial#1#2{{#1  \choose #2}}

\global\long\def\Hspace{\mathcal{F}}
 \global\long\def\Hinner#1#2{\left\langle #1,\,#2\right\rangle _{\Hspace}}

\global\long\def\spin{\sigma}
 \global\long\def\spinconf{\boldsymbol{\sigma}}
 \global\long\def\energy{\varepsilon}
 \global\long\def\contourset{\mathcal{C}}

\global\long\def\rect{R}
 \global\long\def\rectTop{\mathrm{top}}
 \global\long\def\rectBot{\mathrm{bot}}

\global\long\def\NE{\mathrm{NE}}
 \global\long\def\NW{\mathrm{NW}}
 \global\long\def\SW{\mathrm{SW}}
 \global\long\def\SE{\mathrm{SE}}

\global\long\def\iinterval#1{\llbracket\,\!#1\,\!\rrbracket}
 \global\long\def\hiinterval#1{\llbracket\,\!#1\,\!\rrbracket^{*}}

\global\long\def\Srow{\mathrm{Srow}}
 \global\long\def\Cliff{\mathfrak{Cliff}}
 \global\long\def\CliffGen{\Cliff^{(1)}}
 \global\long\def\cre{\boxplus}
 \global\long\def\ann{\boxminus}

\global\long\def\RH{\mathrm{RH}}
 \global\long\def\sP{\mathcal{P}}
 \global\long\def\sS{\mathcal{S}}
 \global\long\def\sW{\mathcal{W}}

\global\long\def\SLE{\mathrm{SLE}}
 \global\long\def\SLEk{\mathrm{SLE}_{\kappa}}

\title{Discrete holomorphicity and Ising model operator formalism}

\author{Clément Hongler, Kalle Kytölä and Ali Zahabi}

\address{Department of Mathematics, Columbia University, 2990 Broadway, New
York, NY 10027, USA}

\email{hongler@math.columbia.edu}

\address{Department of Mathematics and Statistics, P.O. Box 68, FIN\textendash{}00014
University of Helsinki, Finland}

\email{kalle.kytola@helsinki.fi, seyedali.zahabi@helsinki.fi}
\begin{abstract}
We explore the connection between the transfer matrix formalism and
discrete complex analysis approach to the two dimensional Ising model.

We construct a discrete analytic continuation matrix, analyze its
spectrum and establish a direct connection with the critical Ising
transfer matrix. We show that the lattice fermion operators of the
transfer matrix formalism satisfy, as operators, discrete holomorphicity,
and we show that their correlation functions are Ising parafermionic
observables. We extend these correspondences also to outside the critical
point.

We show that critical Ising correlations can be computed with operators
on discrete Cauchy data spaces, which encode the geometry and operator
insertions in a manner analogous to the quantum states in the transfer
matrix formalism.
\end{abstract}
\maketitle
\tableofcontents{}

\section{\label{sec:introduction}Introduction}

The transfer matrix approach to the planar Ising model is both classical
and remarkably powerful \cite{kramers-wannier,baxter,mccoy-wu-i}.
The free energy, critical exponents and a number of correlation functions
of the model were calculated using the transfer matrix, and much of
the algebraic structure underlying the model is easiest understood
by means of the transfer matrix and the related operator formalism.
The formalism is also manifestly suggestive of the quantum field theories
believed to describe the scaling limit of the Ising model.

Recently, methods of discrete complex analysis have lead to significant
progress in the understanding of the Ising model, especially in its
critical phase \cite{smirnov-i,smirnov-iii}. The discrete complex
analysis techniques apply to the model on planar domains of arbitrary
shapes, and allow to prove conformal invariance results.

In this paper, we investigate connections between the transfer matrix
approach and discrete complex analysis techniques. We study discrete
level relations between the two approaches, using concepts of quantum
field theory and analytic tools that are well behaved in the scaling
limit.

\subsection{\label{sub:ising-model}Ising model}

The Ising model describes up/down spins interacting on a lattice.
It is a simple model originally introduced to describe ferromagnetism,
but subsequently it has become a standard in the study of order-disorder
phase transition. 

The Ising model is a random assignment of $\pm1$ spins to the vertices
a graph, that interact via the edges of the graph. We will consider
the Ising model on subgraphs $\mathcal{G}=(\mathcal{V},\mathcal{E})$
of the square lattice $\bZ^{2}$.

The probability of a spin configuration $\left(\mathbf{s}_{x}\right)_{x\in\mathcal{V}}$
is proportional to the Boltzmann weight $e^{-\beta H\left(\mathbf{s}\right)}$,
where $\beta>0$ is the inverse temperature and $H$ is the energy
given by $H(\mathbf{s})=-\sum_{i\sim j}\mathbf{s}_{i}\mathbf{s}_{j}$
with sum over pairs of adjacent vertices $<i,j>\in\mathcal{E}$. Hence,
the model favors local alignment of spins by assigning them a lower
energy, and the strength of this effect is controlled by $\beta$.

For the Ising model in dimensions at least two, a phase transition
in the large scale behavior occurs at a critical inverse temperature,
for the square lattice Ising model at $\beta_{c}=\frac{1}{2}\ln\left(\sqrt{2}+1\right)$.
For $\beta<\beta_{c}$, the system is disordered: spins at large distances
decorrelate, i.e. there is no alignment. For $\beta>\beta_{c}$ the
system has long-range order: spins are uniformly positively correlated,
i.e. global alignment takes place. To properly make sense of the large
scale behavior, one considers either the thermodynamic limit in which
the graph tends to the infinite square lattice $\mathcal{G}\nearrow\bZ^{2}$,
or the scaling limit in which a given planar domain $\Omega$ is approximated
by subgraphs $\Omega_{\delta}$ of $\delta\bZ^{2}$, the square lattice
with fine lattice mesh $\delta\searrow0$.

The square-lattice Ising model is exactly solvable: in particular,
the free energy and thermodynamical properties of the model are well
understood. However, the fine nature of the critical phase and its
precise connection with quantum field theory have for long remained
mysterious from a mathematical perspective. Renormalization group
and quantum field theory methods have provided a non-rigorous insight
into the nature of the phase transition, Conformal Field Theory in
particular giving numerous exact predictions. At the critical point,
$\beta=\beta_{c}$, the model (like many critical two dimensional
lattice models) should have a universal, conformally invariant scaling
limit. Recently some of this insight has become tractable mathematically
by the development of discrete complex analysis techniques: one can
make sense of the scaling limits of the fields and the curves of the
model at the critical temperature. 
\begin{itemize}
\item The scaling limits of the random fields of the model are described
by a Conformal Field Theory. The CFTs are quantum field theories with
infinite-dimensional symmetries, which allow one to compute the critical
exponents and the correlation functions via representation theoretic
methods \cite{belavin-polyakov-zamolodchikov-i,belavin-polyakov-zamolodchikov-ii}.
\item The scaling limits of the random curves of the model are described
by a Schramm-Loewner Evolution. The SLEs are random processes characterized
by their conformal invariance and a Markovian property with respect
to the domain \cite{schramm}.
\end{itemize}
A natural framework to investigate full conformal invariance of the
Ising model (and other models) is to study the model on arbitrary
planar domains, with boundary conditions. A number of results in this
framework has been obtained in recent years: the convergence in the
scaling limit has been shown for parafermionic observables \cite{smirnov-i,smirnov-ii,chelkak-smirnov-ii},
for the energy correlations \cite{hongler-smirnov-ii,hongler-i} and
for the spin correlations \cite{chelkak-izyurov,chelkak-hongler-izyurov}.
These scaling limit results for correlations rely, for a large part,
on discrete complex analysis. They have in turn provided the key tools
to identify and control convergence of the random curves in the scaling
limit \cite{smirnov-i,chelkak-duminil-hongler-kemppainen-smirnov,hongler-kytola}.

In the special case of the full plane, the progress in the study of
scaling limits of Ising model at and near criticality has been steady
over a longer time. In notable results, massive correlations in the
full plane have been computed \cite{wu-mccoy-tracy-barouch} and formulated
in terms of holonomic field theory \cite{sato-miwa-jimbo-i,sato-miwa-jimbo-ii,sato-miwa-jimbo-iii,sato-miwa-jimbo-iv,palmer-tracy,palmer}.
Critical correlations in the full plane have been computed using dimer
techniques and discrete analysis \cite{boutillier-de-tiliere-i,boutillier-de-tiliere-ii,dubedat-i,dubedat-ii}.

\subsection{\label{sub:tm-dhol}Transfer matrix and discrete holomorphicity}

In this subsection, we briefly introduce the two approaches to the
Ising model studied in this paper: the transfer matrix and the discrete
complex analysis formalisms.

\subsubsection{\label{sub:transfer-matrix-approach}Transfer matrix approach}

Let \textbf{$\mathbf{I}$} be an interval of $\mathbb{Z}$, with boundary
$\partial\mathbf{I}\subset\mathbf{I}$ consisting of the two endpoints
of the interval, and consider the rectangular box $\mathbf{I}\times\left\{ 0,\ldots,N\right\} $
with rows $\mathbf{I}_{0},\ldots,\mathbf{I}_{N}$ (set $\mathbf{I}_{y}:=\mathbf{I}\times\left\{ y\right\} $).
Using the transfer matrix, we can represent the Ising model on $\mathbf{I}\times\left\{ 0,\ldots,N\right\} $
as a quantum evolution of spins living on $\mathbf{I}$ from time
$0$ to time $N$.

The Ising model transfer matrix $V$ acts on a state space $\mathcal{S}$,
which has basis $(\mathbf{e}_{\sigma}$) indexed by spin configurations
in a row, $\sigma\in\left\{ \pm1\right\} ^{\mathbf{I}}$. We set $V:=\left(V^{\mathrm{h}}\right)^{\half}V^{\mathrm{v}}\left(V^{\mathrm{h}}\right)^{\half}$,
where the factors separately account for Ising interactions along
horizontal and vertical edges. The matrix element of $V^{\mathrm{v}}$
at $\sigma,\rho\in\left\{ \pm1\right\} ^{\mathbf{I}}$ is defined
(with fixed boundary conditions) as 
\[
V_{\sigma\rho}^{\mathrm{v}}=\begin{cases}
\exp\left(\beta\sum_{i\in\mathbf{I}}\sigma_{i}\rho_{i}\right) & \mbox{ if }\sigma|_{\partial\mathbf{I}}=\rho|_{\partial\mathbf{I}}\\
0 & \mbox{otherwise}
\end{cases}
\]
 and the matrix $V^{\mathrm{h}}$ is diagonal with elements 
\[
V_{\sigma\sigma}^{\mathrm{h}}=\exp\left(\beta\sum_{x\sim y}\sigma_{x}\sigma_{y}\right).
\]

Viewing the $y$-axis as time, the transfer matrix can be thought
of as an exponentiated quantum Hamiltonian in $1+1$ dimensional space-time:
at the row $\mathbf{I}_{y}$, we have a quantum state $\mathbf{v}_{y}\in\mathcal{S}$
which we propagate to the next row $\mathbf{I}_{y+1}$ by $\mathbf{v}_{y+1}=V\mathbf{v}_{y}$. 

In the path integral picture, the evolution $\mathbf{v}_{0},\ldots,\mathbf{v}_{N}$
becomes a sum over trajectories weighted by their amplitudes: the
trajectories are spin configurations $\mathbf{s}:\mathbf{I}\times\left\{ 0,\ldots,N\right\} \to\left\{ \pm1\right\} $
and the amplitudes are the Boltzmann weights $e^{-\beta H\left(\mathbf{s}\right)}$.
As a result, the partition function $Z=\sum_{\mathbf{s}}e^{-\beta H\left(\mathbf{s}\right)}$
equals $\left\langle \mathbf{f}|V^{N}|\mathbf{i}\right\rangle :=\mathbf{f}^{\top}V^{N}\mathbf{i}$,
where $\mathbf{i},\mathbf{f}\in\mathcal{S}$ encode the boundary conditions
on $\mathbf{I}_{0}$ and $\mathbf{I}_{N}$. 

Ising fields (such as the spin, energy, disorder, fermions) are represented
by the insertion of corresponding operators. The position of the fields
appear in two ways: the operator is applied to the state on the row
$y$ on which the field lives, and the applied operator $O_{x}:\mathcal{S}\to\mathcal{S}$
depends on the position $x$ of the field in that row. We combine
the dependence on the horizontal coordinate $x$ and the vertical
time coordinate $y$ by using the operator $O(z)=V^{-y}O_{x}V^{y}$
for the field located at $z=x+\ii y$. Then the correlation function
of fields $O^{(1)},\ldots,O^{(n)}$ located at $z_{1},\ldots,z_{n}$
is
\[
\left\langle O^{(1)}\left(z_{1}\right)\cdots O^{(n)}\left(z_{n}\right)\right\rangle =\frac{\langle\mathbf{f}|V^{N}O^{(1)}\left(z_{1}\right)\cdots O^{(n)}\left(z_{n}\right)|\mathbf{i}\rangle}{\langle\mathbf{f}|V^{N}|\mathbf{i}\rangle}.
\]
For probabilistic fields such as the spin $\mathbf{s}_{z}$, represented
by $\hat{\sigma}\left(z\right):\mathcal{S}\to\mathcal{S}$, the correlation
functions are the expected values of products, e.g. $\left\langle \hat{\sigma}\left(z_{1}\right)\cdots\hat{\sigma}\left(z_{n}\right)\right\rangle =\mathbb{E}\left[\mathbf{s}_{z_{1}}\cdots\mathbf{s}_{z_{n}}\right]$.
Non-probabilistic fields (such as fermion and disorder) can also be
represented within the transfer matrix formalism.

Also in Conformal Field Theory, the field-to-operator correspondence
is fundamental. However, naively connecting the algebraic structure
of Ising model and the one of CFT is problematic: the transfer matrix
does not have a nice scaling limit and it is best suited to very specific
geometries (rectangle, cylinder, torus, plane).

Contrary to the transfer matrix formalism, discrete complex analysis
is well suited to handle scaling limits on domains of arbitrary geometry,
and hence to discuss conformal invariance. For this reason, relating
the transfer matrix to discrete complex analysis seems a promising
way to provide a manageable scaling limit for the quantum field theoretic
concepts of the transfer matrix formalism.

\subsubsection{\label{sub:discrete-complex-analysis}Discrete complex analysis approach}

The idea of discrete complex analysis is to identify fields on lattice
level, whose correlations satisfy difference equations --- lattice
analogues of equations of motion. A particularly useful type of such
equations are strong lattice Cauchy-Riemann equations (massless at
$\beta_{c}$, massive at $\beta\neq\beta_{c}$), which we will refer
to as '\emph{s-holomorphicity}'. 

For the critical Ising model, certain s-holomorphic fields can be
completely characterized in terms of discrete complex analysis: their
correlation functions (called 'observables') can be formulated as
the unique solutions to discrete Riemann-type boundary value problems
(RBVP). The convergence of s-holomorphic observables is in particular
the main tool to establish convergence of Ising interfaces to SLE
\cite{smirnov-i,chelkak-duminil-hongler-kemppainen-smirnov,hongler-kytola},
and to prove conformal invariance of the energy and the spin correlations
\cite{hongler-smirnov-ii,hongler-i,chelkak-hongler-izyurov}.

A key example is the Ising parafermionic observable of \cite{chelkak-smirnov-ii}.
On a discrete domain $\Omega$ (finite simply connected union of faces
of $\mathbb{Z}^{2}$), for two midpoints of edges $a$ and $z$, the
observable is defined by 
\[
f\left(a,z\right)=\frac{1}{\mathcal{Z}}\sum_{\gamma:a\leadsto z}e^{-2\beta\#\mathrm{edges}\left(\gamma\right)}e^{-\frac{i}{2}\mathrm{winding}\left(\gamma:a\to z\right)},
\]
where $\mathcal{Z}$ is a partition function, the sum is over collections
$\gamma$ of dual edges, consisting of loops and a path from $a$
to $z$, and $\mathrm{winding}\left(\gamma:a\to z\right)$ is the
total turning angle of the path.

At critical temperature $\beta=\beta_{c}$, when $a$ is a bottom
horizontal boundary edge, the function $f_{a}:=f(a,\cdot)$ is the
unique solution of a discrete RBVP:
\begin{itemize}
\item $f_{a}$ is s-holomorphic: for any two incident edges $e_{v}=\left\langle vu\right\rangle $
and $e_{w}=\left\langle wu\right\rangle $, the values of $f_{a}$
satisfy the real-linear equation $f_{a}\left(e_{v}\right)+\frac{i}{\theta}\overline{f_{a}\left(e_{v}\right)}=f_{a}\left(e_{w}\right)+\frac{i}{\theta}\overline{f_{a}\left(e_{w}\right)}$,
where $\theta=\frac{2u-v-w}{\left|2u-v-w\right|}$. 
\item On the boundary, values of $f_{a}$ are real multiples of $\tau_{\mathrm{cw}}^{-\frac{1}{2}}$,
where $\tau_{\mathrm{cw}}$ is the clockwise tangent to the boundary. 
\item $f_{a}$ satisfies the normalization condition $f_{a}(a)=1$.
\end{itemize}
One can then show that the solutions of discrete RBVPs converge to
the solutions of continuous RBVPs, which are conformally covariant
\cite{smirnov-i,smirnov-ii,hongler-i,hongler-smirnov-ii,chelkak-smirnov-ii,chelkak-izyurov,chelkak-hongler-izyurov}.

The approach of s-holomorphic functions has proved succesful for the
study of conformal invariance: it applies to arbitrary planar geometries,
general graphs and behaves well in the scaling limit. Still, the algebraic
structures of CFT are not apparent in the s-holomorphic approach:
there is no Hilbert space of states, no obvious action of the Virasoro
algebra and no simple reason for the continuous correlations to obey
the CFT null-field PDEs. To connect the Ising model with CFT, one
would like to write algebraic data (e.g. from transfer matrix) in
s-holomorphic terms and then to pass to the limit.

\subsection{\label{sub:main-results}Main results}

The goal of this paper is to explore the connection between the transfer
matrix formalism and s-holomorphicity approach to the critical Ising
model, and to lay foundations for a quantum field theoretic description
that behaves well in the scaling limit.

We construct a discrete analytic continuation matrix, analyze its
spectrum and establish a direct connection with the Ising transfer
matrix. We show that the lattice fermion operators of the transfer
matrix formalism satisfy, as operators, s-holomorphic equations of
motion, and we show that their correlation functions are s-holomorphic
Ising parafermionic observables. Finally, we show that Ising correlations
can be computed with lattice Poincaré-Steklov operators, which encode
the geometry and operator insertions in a manner analogous to the
quantum states in the transfer matrix formalism.

The results admit generalizations to non-critical Ising model, with
s-holomorphicity replaced by a concept of massive s-holomorphicity.

\subsubsection{\label{sub:intro-ising-tm-shol}Discrete analytic continuation and
Ising transfer matrix}

Let $a<b$ be integers, consider the interval $\mathbf{I}:=\left\{ x\in\mathbb{Z}:a\leq x\leq b\right\} $
and let $\partial\mathbf{I}:=\left\{ a,b\right\} $ denote its boundary.
Let $\mathbf{I}^{*}$ be the dual of $\mathbf{I}$, the set of half-integers
between $a$ and $b$. Write $\mathbf{I}_{0}^{*}:=\mathbf{I}^{*}\times\left\{ 0\right\} $,
$\mathbf{I}_{\frac{1}{2}}:=\mathbf{I}\times\left\{ \frac{1}{2}\right\} $,
etc. For simplicity of notation, we identify edges with their midpoints.
\begin{lem*}[Section \ref{sub:s-hol-prop}]
\label{lem:existence-propagation} Let $f:\mathbf{I}_{0}^{*}\to\mathbb{C}$
be a complex-valued function. Then there is a unique s-holomorphic
extension $h$ of $f$ to $\mathbf{I}_{0}^{*}\cup\mathbf{I}_{\frac{1}{2}}\cup\mathbf{I}_{1}^{*}$
with Riemann boundary values on $\partial\mathbf{I}_{\frac{1}{2}}$. 
\end{lem*}
Since s-holomorphicity and RBVP are $\bR$-linear concepts, we identify
$\mathbb{C}\isom\mathbb{R}^{2}$ and denote by $P:\left(\mathbb{R}^{2}\right)^{\mathbf{I}^{*}}\to\left(\mathbb{R}^{2}\right)^{\mathbf{I}^{*}}$
the $\bR$-linear linear map $f\mapsto h\Big|_{\mathbf{I}_{1}^{*}}$.
In other words, $P$ is the row-to-row propagation of s-holomorphic
solutions of the Riemann boundary value problem.
\begin{prop*}[Proposition \ref{prop:form-of-the-eigenvalues} in Section \ref{sub:diag-prop}.]
 The operator $P$ can be diagonalized and has a positive spectrum,
given by $\lambda_{\alpha}^{\pm1}$ where $\lambda_{\alpha}>1$ are
distinct for $\alpha=1,\ldots,\left|\mathbf{I}\right|^{*}$.
\end{prop*}
Let $P^{\mathbb{C}}:\left(\mathbb{C}^{2}\right)^{\mathbf{I}^{*}}\to\left(\mathbb{C}^{2}\right)^{\mathbf{I}^{*}}$
be the complexification of $P$, i.e. the $\mathbb{C}$-linear map
such that $P^{\mathbb{C}}\big|_{\left(\mathbb{R}^{2}\right)^{\mathbf{I}^{*}}}=P$.
Let $W_{\circ}\subset\left(\mathbb{C}^{2}\right)^{\mathbf{I}^{*}}$
be the vector space spanned by the eigenvectors of $P^{\mathbb{C}}$
of eigenvalues less than $1$ and let $P_{\circ}^{\mathbb{C}}:W_{\circ}\to W_{\circ}$
be the restriction of $P^{\mathbb{C}}$ to $W_{\circ}$. 
\begin{thm*}[Theorem \ref{thm:transfer-matrix-on-physical-Fock-space} in Section
\ref{sub:fock-representations}]
\label{thm:tm-fock-s-hol} Let $\bigwedge W_{\circ}$ be the exterior
tensor algebra $\bigoplus_{n=0}^{\left|\mathbf{I}^{*}\right|}\bigwedge^{n}W_{\circ}$
and let $\Gamma\left(P_{\circ}^{\mathbb{C}}\right):\bigwedge W_{\circ}\to\bigwedge W_{\circ}$
be defined as $\bigoplus_{n=0}^{\left|\mathbf{I}^{*}\right|}\left(P_{\circ}^{\mathbb{C}}\right)^{\otimes n}$.
Let $V_{+}:\mathcal{S}_{+}\to\mathcal{S}_{+}$ be the Ising model
transfer matrix at the critical point $\beta=\beta_{c}$, restricted
to the subspace $\mathcal{S}_{+}\subset\mathcal{S}$ defined as $\spn\left\{ \mathbf{e}_{\sigma}:\sigma_{b}=1\right\} $
(see Section \ref{sub:transfer-matrix-approach}).

Then there is an isomorphism $\rho:\mathcal{S}_{+}\to\bigwedge\left(W_{\circ}\right)$
such that $\rho\circ V\circ\rho^{-1}=\Lambda_{0}\times\Gamma\left(P_{\circ}^{\mathbb{C}}\right)$
for some $\Lambda_{0}>0$.
\end{thm*}
It follows in particular that the spectrum of the critical Ising model
transfer matrix is completely determined by the spectrum of the discrete
analytic continuation matrix $P$.

\subsubsection{\label{sub:intro-induced-rotation}Induced rotation and s-holomorphic
propagation}

The theorem of Section \ref{thm:tm-fock-s-hol} relies on the Kaufman
representation of the Ising transfer matrix \cite{kaufman}: $V$
can be constructed from its so-called induced rotation $T_{V}$ on
a space of Clifford generators defined below. The connection with
discrete analysis is made by observing that the s-holomorphic propagation
$P^{\mathbb{C}}$ is actually equal (up to a change of basis) to $T_{V}$. 

For $k\in\mathbf{I}^{*}$ and a spin configuration $\sigma\in\left\{ \pm1\right\} ^{\mathbf{I}}$.
We define the operators $p_{k}:\mathcal{S}\to\mathcal{S}$ and $q_{k}:\mathcal{S}\to\mathcal{S}$
by
\begin{align*}
 & \begin{array}{c}
p_{k}\left(\mathbf{e}_{\sigma}\right)=\;\sigma_{k+\frac{1}{2}}\mathbf{e}_{\tau}\\
q_{k}\left(\mathbf{e}_{\sigma}\right)=\ii\,\sigma_{k-\frac{1}{2}}\mathbf{e}_{\tau}
\end{array}\qquad\text{, where} & \tau_{x}=\; & \begin{cases}
\phantom{-}\sigma_{x\quad} & \text{for }x>k\\
-\sigma_{x}\quad & \text{for }x<k.
\end{cases}
\end{align*}
Let $\mathcal{W}$ be the space of operators $\mathcal{S}\to\mathcal{S}$
spanned by $\left\{ p_{k},q_{k}\,\big|\, k\in\mathbf{I}^{*}\right\} $.
The conjugation $\mathcal{O}\mapsto V\mathcal{O}V^{-1}$ defines a
linear operator $\mathcal{W}\to\mathcal{W}$, which we denote by $T_{V}$
and call the \emph{induced rotation} of $V$.
\begin{thm*}[Theorem \ref{thm:massive-s-hol-prop-induced-rotation} in Section
\ref{sub:induced-rotation}]
Let $T_{V}:\mathcal{W}\to\mathcal{W}$ be the induced rotation of
$V$ at critical point $\beta=\beta_{c}$, and let $P^{\mathbb{C}}:\left(\mathbb{C}^{2}\right)^{\mathbf{I}^{*}}\to\left(\mathbb{C}^{2}\right)^{\mathbf{I}^{*}}$
be the complexified s-holomorphic propagation. Then there is an isomorphism
$\varrho:\left(\mathbb{C}^{2}\right)^{\mathbf{I}^{*}}\to\mathcal{W}$
such that $T_{V}:=\varrho\circ P^{\mathbb{C}}\circ\varrho^{-1}$.
\end{thm*}

\subsubsection{\label{sub:intro-fermion-operators}Fermion operators}

An important tool for the analysis of the Ising model in the transfer
matrix formalism are the fermion operators; similarly, the study of
the scaling limit of the Ising model on planar domains relies on s-holomorphic
parafermionic observables. We discuss two facts pertaining to the
relation of these two, namely that the fermion operators are complexified
s-holomorphic (as matrix-valued functions) and that their correlations
are indeed the parafermionic observables discussed in Section \ref{sub:intro-ising-tm-shol}.
\begin{thm*}[Theorem \ref{thm:local-relations-for-the-fermion-operators} in Section
\ref{sub:s-hol-fermion-op}]
 For $x\in\mathbf{I}^{*}$, define the fermion operators $\psi_{x},\bar{\psi}_{x}:\mathcal{S}\to\mathcal{S}$
by $\psi_{x}:=\frac{\ii}{\sqrt{2}}\left(p_{x}+q_{x}\right)$ and $\bar{\psi}_{x}:=\frac{1}{\sqrt{2}}\left(p_{x}-q_{x}\right)$.
Define the operator-valued fermions on horizontal edges $x+\ii y\in\mathbf{I}^{*}\times\mathbf{J}$
by $\psi\left(x+\ii y\right)=V^{-y}\psi_{x}V^{y}$ and $\bar{\psi}\left(x+\ii y\right)=V^{-y}\bar{\psi}_{x}V^{y}$.
At the critical point $\beta=\beta_{c}$, the pair $\left(\psi,\bar{\psi}\right)$
has a unique operator-valued extension to the edges of $\mathbf{I}\times\mathbf{J}$,
which satisfies complexified s-holomorphic equations (see Section
\ref{sub:s-hol-fermion-op}).
\end{thm*}
Conversely, the s-holomorphic parafermionic observables of \cite{smirnov-ii,chelkak-smirnov-ii,hongler-smirnov-ii,hongler-i}
are indeed correlation functions of the fermion operators.
\begin{thm*}[Theorems \ref{thm:fermion-operator-two-point-functions} and \ref{thm:multipoint-operator-correlations-and-observables}
in Sections \ref{sub:winding-observables-and-low-T-expansions} and
\ref{sub:Multipoint-parafermionic-observables}]
 The correlation functions of the fermion operators are linear combinations
of s-holomorphic parafermionic observables. In particular, in the
box $\mathbf{I}\times\mathbf{J}$, in the setup of Section \ref{sub:discrete-complex-analysis},
we have 
\[
\left\langle \psi\left(z\right)\bar{\psi}\left(a\right)\right\rangle =f\left(a,z\right).
\]
More generally, all the multi-point correlation functions of $\psi$
and $\bar{\psi}$ can be written in terms of parafermionic observables.
\end{thm*}
This allows one to combine the algebraic content carried by the transfer
matrix formalism with the analytic content of the s-holomorphicity
formalism. As an application we give a simple general proof of the
Pfaffian formulas for the multi-point parafermionic observables, transparently
based on the fermionic Wick's formula.

\subsubsection{Operators on Cauchy data spaces}

The above results relate the transfer matrix formalism, close in spirit
to Conformal Field Theory, and s-holomorphicity, suited for scaling
limits and conformal invariance. We would like to interpret some of
the content of the transfer matrix structure in s-holomorphic terms.
The goal is to pass to the scaling limit and to connect the model
with CFT. Can we construct quantum states in s-holomorphic terms,
that encode domain geometry and insertions, and have a scaling limit? 

We present an algebraic construction that encodes the geometry of
a domain in a Poincaré-Steklov operator: all the relevant information
about the domain (for correlations) is contained in the operator.
This operator converges to a bounded singular integral operator in
the scaling limit.

Let $\Omega$ be a square grid domain with edges $\mathcal{E}$, let
$\mathfrak{b}\subset\partial\mathcal{E}$ be a collection of boundary
edges. Let $\mathcal{R}_{\Omega}^{\mathfrak{b}}$ (resp. $\mathcal{I}_{\Omega}^{\mathfrak{b}}$)
be the Cauchy data space of functions $f:\mathfrak{b}\to\mathbb{C}$
such that $f\parallel\tau_{\mathrm{ccw}}^{-\frac{1}{2}}$ on $\mathfrak{b}$
(resp. $f\parallel\tau_{\mathrm{cw}}^{-\frac{1}{2}}$ on $\mathfrak{b}$),
where $\tau_{\mathrm{ccw}}=-\tau_{\mathrm{cw}}$ is the counterclockwise
tangent to $\partial\Omega$.
\begin{lem*}[Lemma \ref{lem:rps-existence} in Section \ref{sub:Discrete-RPS-operators}]
 For any $u\in\mathcal{R}_{\Omega}^{\mathfrak{b}}$, there exists
a unique $v\in\mathcal{I}_{\Omega}^{\mathfrak{b}}$ such that $u+v$
has an s-holomorphic extension $h:\Omega\to\mathbb{C}$ satisfying
$h\parallel\tau_{\mathrm{cw}}^{-\frac{1}{2}}$ on $\partial\Omega\setminus\mathfrak{b}$.
The mapping $u\mapsto v$ defines a real-linear isomorphism $U_{\Omega}^{\mathfrak{b}}:\mathcal{R}_{\Omega}^{\mathfrak{b}}\to\mathcal{I}_{\Omega}^{\mathfrak{b}}$. 
\end{lem*}
The operator $U_{\Omega}^{\mathfrak{b}}$ is a discrete Riemann Poincaré-Steklov
operator. The continuous version of this operator is defined and studied
in \cite{hongler-phong}.

When $\Omega=\mathbb{Z}\times\mathbb{Z}_{+}$ and $\mathfrak{b}=\mathbb{Z}\times\left\{ 0\right\} $,
we have $\mathcal{R}_{\Omega}^{\mathfrak{b}}=\mathbb{R}^{\mathbb{Z}}$
and the operator $U_{\Omega,\mathfrak{b}}$ (limit from bounded domains)
is a discrete analogue of the Hilbert transform (the Hilbert transform
maps a function $u:\mathbb{R}\to\mathbb{R}$ to $v:\mathbb{R}\to\ii\mathbb{R}$
such that $u+v$ has a holomorphic extension $\mathbb{H}\to\mathbb{C}$). 
\begin{prop*}[Lemmas \ref{lem:rps-operator-as-conv} and \ref{lem:ps-with-tm} in
Section \ref{sub:Discrete-RPS-operators}]
 The operator $U_{\Omega}^{\mathfrak{b}}$ is a convolution operator,
whose convolution kernel is the Ising parafermionic observable at
the critical point $\beta=\beta_{c}$. When $\Omega=\mathbf{I}\times\left\{ 0,\cdots,N\right\} $
and $\mathfrak{b}=\mathbf{I}_{0}$, then $U_{\Omega}^{\mathfrak{b}}$
is given in terms of the s-holomorphic propagator $P^{N}$.
\end{prop*}
The operators $U_{\Omega}^{\mathfrak{b}}$ can be used to compute
correlation functions by gluing Cauchy data. Denote by $f_{\Omega}(x,y)$
the Ising parafermionic observable in domain $\Omega$, defined as
in Section \ref{sub:discrete-complex-analysis}.
\begin{thm*}[Theorem \ref{thm:pairing-fermions} in Section \ref{sub:fermion-correlations}]
 Let $\Omega_{1},\Omega_{2}$ be two square grid domains with disjoint
interiors, with edges $\mathcal{E}_{1},\mathcal{E}_{2}$, and let
$\mathfrak{b}:=\partial\mathcal{E}_{1}\cap\partial\mathcal{E}_{2}$.
The inverse operator $Q=\left(\id-U_{\Omega_{1}}^{\mathfrak{b}}U_{\Omega_{2}}^{\mathfrak{b}}\right)^{-1}$
exists. For any $x\in\partial\mathcal{E}_{1}\setminus\mathfrak{b}$
and any $y\in\mathcal{E}_{2}$, the critical Ising parafermionic observable
in $\Omega=\Omega_{1}\cup\Omega_{2}$ can be written as
\begin{eqnarray*}
f_{\Omega}\left(x,y\right) & = & \sum_{k,\ell\in\mathfrak{b}}f_{\Omega_{1}}\left(k,x\right)Q_{k,\ell}f_{\Omega_{2}}\left(\ell,y\right).
\end{eqnarray*}

\end{thm*}
In other words, the operator $Q$ allows one to 'glue' the domain
$\Omega_{2}$ to $\Omega_{1}$, and to compute the fermion correlations
on $\Omega_{1}\cup\Omega_{2}$: all the information about each domain
is contained in $U_{\Omega_{1}}^{\mathfrak{b}}$ and $U_{\Omega_{2}}^{\mathfrak{b}}$.

\subsubsection{\label{sub:away-from-critical}Away from critical temperature}

All the results generalize to temperatures other than the critical
one. The fermions of Section \ref{sub:discrete-complex-analysis}
satisfy the same boundary conditions and are \emph{massive }s-holomorphic
(see Section \ref{sub:massive-shol} for definition). A massive s-holomorphic
propagation $P_{\beta}:\left(\mathbb{R}^{2}\right)^{\mathbf{I}^{*}}\to\left(\mathbb{R}^{2}\right)^{\mathbf{I}^{*}}$(see
Section \ref{sub:s-hol-prop}) and the non-critical transfer matrix
are related like in the critical case.
\begin{thm*}
Let $\beta\neq\beta_{c}$. The massive propagator $P_{\beta}$ is
diagonalizable, with distinct eigenvalues $\lambda_{\alpha}^{\pm1}$
with $\lambda_{\alpha}>1$ for $\alpha=1,2,\ldots,|\mathbf{I}^{*}|$.

Theorems of Sections \ref{sub:intro-ising-tm-shol}, \ref{sub:intro-induced-rotation}
and \ref{sub:intro-fermion-operators} hold true, if one considers
the Ising transfer matrix at temperature $\beta$, massive holomorphicity
equations, and the massive s-holomorphic propagation matrix $P_{\beta}$. 
\end{thm*}

\section{\label{sec:s-holomorphicity}S-Holomorphicity and Riemann Boundary
values}

\subsection{\label{sub:s-hol-eqs}S-holomorphicity equations }

S-holomorphicity is a notion of discrete holomorphicity for complex-valued
functions defined on so-called isoradial graphs \cite{chelkak-smirnov-i}.
In this paper, we consider the case of the square lattice: we consider
functions defined on square grid domains, by which we mean a finite
simply connected union of faces of $\mathbb{Z}^{2}$. More precisely,
we will consider functions defined on the \emph{edges} of square grid
domains; when necessary, we will identify these edges with their midpoints.

S-holomorphicity is a real-linear condition on the values of a function
at incident edges; it implies classical discrete holomorphicity (i.e.
lattice Cauchy-Riemann equations) but is strictly stronger. The fact
that Ising model parafermionic observables are s-holomorphic is the
key to establish their convergence in the scaling limit and hence
to prove conformal invariance results.

\begin{figure}

\includegraphics[width=5cm]{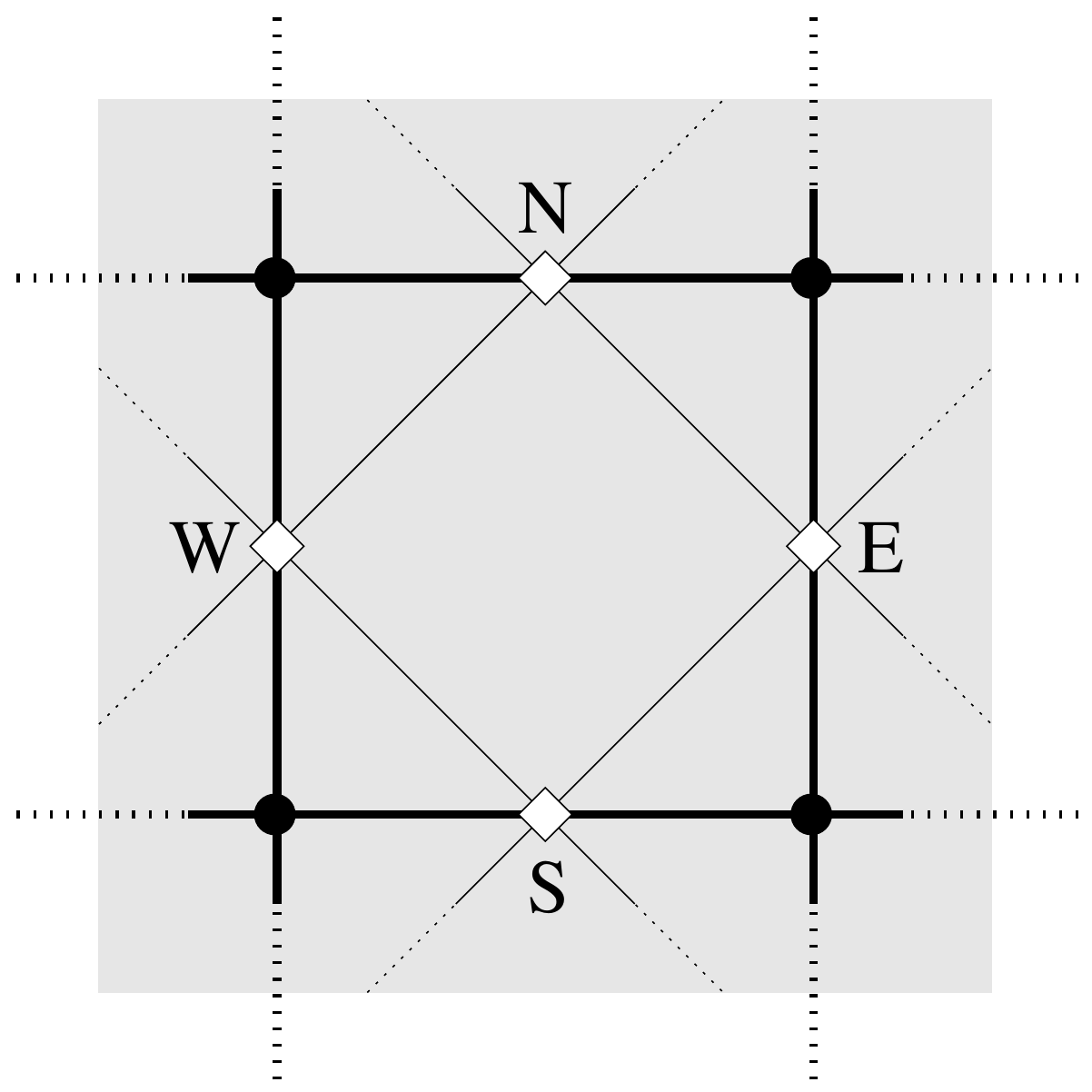}

\caption{\label{fig:The-four-edges}The four edges adjacent to a face on the
square lattice.}

\end{figure}

\begin{defn}
\label{def:s-holomorphic}Let $\Omega$ be a square grid domain. Set
$\lambda:=e^{\ii\pi/4}$. We say that $F\colon\Omega\to\mathbb{C}$
is s-holomorphic if for any face of $\Omega$ with edges $E,N,W,S$
(see Figure \ref{fig:The-four-edges}), the following s-holomorphicity
equations hold

\begin{align}
F(N)+\lambda\overline{F(N)}=\; & F(E)+\lambda\overline{F(E)}\label{eq: s-holomorphicity relations}\\
F(N)+\lambda^{-1}\overline{F(N)}=\; & F(W)+\lambda^{-1}\overline{F(W)}\nonumber \\
F(S)+\lambda^{3}\overline{F(S)}=\; & F(E)+\lambda^{3}\overline{F(E)}\nonumber \\
F(S)+\lambda^{-3}\overline{F(S)}=\; & F(W)+\lambda^{-3}\overline{F(W)}.\nonumber 
\end{align}

\end{defn}
In other words $F\colon\Omega\rightarrow\bC$ is \emph{s-holomorphic}
if for any pair of incident edges $e_{v}=\left\langle uv\right\rangle $
and $e_{w}=\left\langle uw\right\rangle $, we have $F\left(e_{v}\right)+\frac{\ii}{\theta}\overline{F\left(e_{v}\right)}=F\left(e_{w}\right)+\frac{\ii}{\theta}\overline{F\left(e_{w}\right)}$,
where $\theta=\frac{2u-v-w}{\left|2u-v-w\right|}$. Equivalently,
the orthogonal projections (in the complex plane) of $F\left(e_{v}\right)$
and $F\left(e_{w}\right)$ on the line $\sqrt{\frac{i}{\theta}}\mathbb{R}$
coincide.

The above equations imply (but are not equivalent to) the usual lattice
Cauchy-Riemann equations: for the four edges around a face as in Figure
\ref{fig:The-four-edges} we have $F\left(N\right)-F\left(S\right)=\ii\left(F\left(E\right)-F\left(W\right)\right)$,
and a similar equation holds for the four edges incident to a vertex.
Discrete Cauchy-Riemann equations imply in turn the discrete Laplace
equation $\sum_{Z=X\pm1,X\pm i}\left(F\left(Z\right)-F\left(X\right)\right)=0$
for every edge $X\in\Omega\setminus\partial\Omega$.

\subsection{\label{sub:massive-shol}Massive s-holomorphicity}

We now define a perturbation of s-holomorphicity which we call massive
s-holomorphicity. The massive s-holomorphicity equations with parameter
$\beta$ are $\bR$-linear equations satisfied by the Ising model
parafermionic observables at inverse temperature $\beta$. At the
critical point $\beta=\beta_{c}$, massive s-holomorphicity reduces
to s-holomorphicity. 
\begin{defn}
\label{def:massive-s-hol}Let $\beta>0$, let $\nu=\nu\left(\beta\right)$
be the unit complex number be defined by $\nu=\overline{\lambda}^{3}\frac{\alpha+i}{\alpha-i}$,
where $\alpha=e^{-2\beta}$ and $\lambda=e^{\ii\pi/4}$. A function
$F\colon\Omega\rightarrow\bC$ is said to be \emph{massive s-holomorphic
with parameter $\beta$} if for any face of $\Omega$ with edges $E,N,W,S$,
we have

\begin{align}
F(N)+\nu^{-1}\lambda\overline{F(N)}=\; & \nu^{-1}F(E)+\lambda\overline{F(E)}\label{eq: massive s-holomorphicity relations}\\
F(N)+\nu\lambda^{-1}\overline{F(N)}=\; & \nu F(W)+\lambda^{-1}\overline{F(W)}\nonumber \\
F(S)+\nu\lambda^{3}\overline{F(S)}=\; & \nu F(E)+\lambda^{3}\overline{F(E)}\nonumber \\
F(S)+\nu^{-1}\lambda^{-3}\overline{F(S)}=\; & \nu^{-1}F(W)+\lambda^{-3}\overline{F(W)}.\nonumber 
\end{align}

\end{defn}
At $\beta=\beta_{c}$, we have $\nu=1$ and these equations coincide
with the Equations (\ref{eq: s-holomorphicity relations}) defining
s-holomorphicity. It can be shown (see \cite{beffara-duminil}) that
massive s-holomorphicity implies (but is not equivalent to) the massive
Laplace equation 
\[
\frac{1}{4}\sum_{Z=X\pm1,X\pm i}\left(F\left(Z\right)-F\left(X\right)\right)=\mu F\left(X\right)\quad\forall X\in\Omega\setminus\partial\Omega,
\]
with the mass $\mu=\mu(\beta)$ given by $\mu=\frac{S+S^{-1}}{2}-1$,
where $S=\sinh\left(2\beta\right)$. The dual inverse temperatures
$\beta$ and $\beta^{*}$, related by $\sinh(2\beta)\sinh(2\beta^{*})=1$,
have equal masses $\mu(\beta)=\mu(\beta^{*})$, and at the critical
point the mass vanishes $\mu(\beta_{c})=0$.

\subsection{\label{sub:riemann-boundary-value}Riemann boundary values}

The boundary conditions that are relevant for the study of Ising model
specify the argument of a function on the boundary edges $\partial\Omega$:
these conditions are trivially satisfied by the Ising parafermionic
observables for topological reasons (see Section \ref{sub:winding-observables-and-low-T-expansions}),
at any temperature.

Let $\Omega$ be a discrete square grid domain. The boundary $\partial\Omega$
of $\Omega$ is a simple closed curve. For an edge $e\in\partial\Omega$,
this defines a clockwise orientation $\tau_{\mathrm{cw}}\left(e\right)$
of $e$, which we view as a complex number: $\tau_{\mathrm{cw}}\left(e\right)\in\left\{ \pm1\right\} $
if $e$ is horizontal and $\tau_{\mathrm{cw}}\left(e\right)=\left\{ \pm\ii\right\} $
if $e$ is vertical.
\begin{defn}
\noindent \label{def:riemann-bc}We say that a function $f:\Omega\to\mathbb{C}$
satisfies Riemann boundary conditions $f\parallel\tau_{\mathrm{cw}}^{-\frac{1}{2}}$
at an edge $z\in\partial\Omega$ if 
\[
f\left(z\right)\parallel\tau_{\mathrm{cw}}^{-\frac{1}{2}}\left(z\right),
\]
i.e. $f\left(z\right)$ is a real multiple of $\tau_{\mathrm{cw}}^{-\half}\left(z\right)$.
\end{defn}
When $\Omega$ is a rectangular box $\mathbf{I}\times\mathbf{J}$,
the condition $f\parallel\tau_{\mathrm{cw}}^{-\frac{1}{2}}$ means
that $f$ is purely real on the top side of $\Omega$, purely imaginary
on the bottom side, a real multiple of $\lambda=e^{\ii\pi/4}$ on
the left side and a real multiple of $\lambda^{-1}=e^{-\ii\pi/4}$
on the right side.

\subsection{\label{sub:s-hol-prop}S-holomorphic continuation operator}

For a (massive) s-holomorphic function on a rectangular box with Riemann
boundary conditions $\parallel\tau_{\mathrm{cw}}^{-\frac{1}{2}}$,
we can propagate its values row by row as illustrated in Figure \ref{fig:Propagations}.
This is supplied by the following lemma (we use the same notation
as in Section \ref{sub:intro-ising-tm-shol}). 
\begin{lem}
\label{lem:massive-shol-prop}Consider the box $\mathbf{I}\times\left\{ 0,1\right\} $
for an integer interval $\mathbf{I}=\left[a,b\right]\cap\mathbb{Z}$
and let $\mathbf{I^{*}}=[a,b]\cap(\bZ+\half)$ be its dual. 

Let $f:\mathbf{I}_{0}^{*}\to\mathbb{C}$ be a complex-valued function
and let $\beta>0$. Then there is a unique massive s-holomorphic extension
$h$ of $f$ to $\mathbf{I}_{0}^{*}\cup\mathbf{I}_{\frac{1}{2}}\cup\mathbf{I}_{1}^{*}$
with Riemann boundary values on $\partial\mathbf{I}_{\frac{1}{2}}$.\end{lem}
\begin{proof}
For $z\in\mathbf{I}_{\frac{1}{2}}\setminus\partial\mathbf{I}_{\frac{1}{2}}$,
the value $h\left(z\right)$ can be solved uniquely from the last
two of Equations (\ref{eq: massive s-holomorphicity relations}) in
terms of $f\left(z-\frac{1}{2}-\frac{i}{2}\right)$ and $f\left(z+\frac{1}{2}-\frac{i}{2}\right)$.
For $z\in\partial\mathbf{I}_{\frac{1}{2}}$, the value $h\left(z\right)$
can be solved uniquely from the Riemann boundary condition and (\ref{eq: massive s-holomorphicity relations})
in terms the value $f\left(z-\frac{i}{2}\pm\frac{1}{2}\right)$ ($\pm$
depending on whether $z$ is on the left or the right part of $\partial\mathbf{I}_{\frac{1}{2}}$).
For $z\in\mathbf{I}_{1}^{^{*}}$, $h\left(z\right)$ can be solved
in terms of $h$ at $z+\frac{1}{2}-\frac{i}{2}\in\mathbf{I}_{\frac{1}{2}}$
and $z-\frac{1}{2}-\frac{i}{2}\in\mathbf{I}_{\frac{1}{2}}$ by the
first two of Equations (\ref{eq: massive s-holomorphicity relations}).
The definition of $h$ thus obtained satisfies all the required equations.
\end{proof}
\begin{figure}
\includegraphics[width=8.95cm]{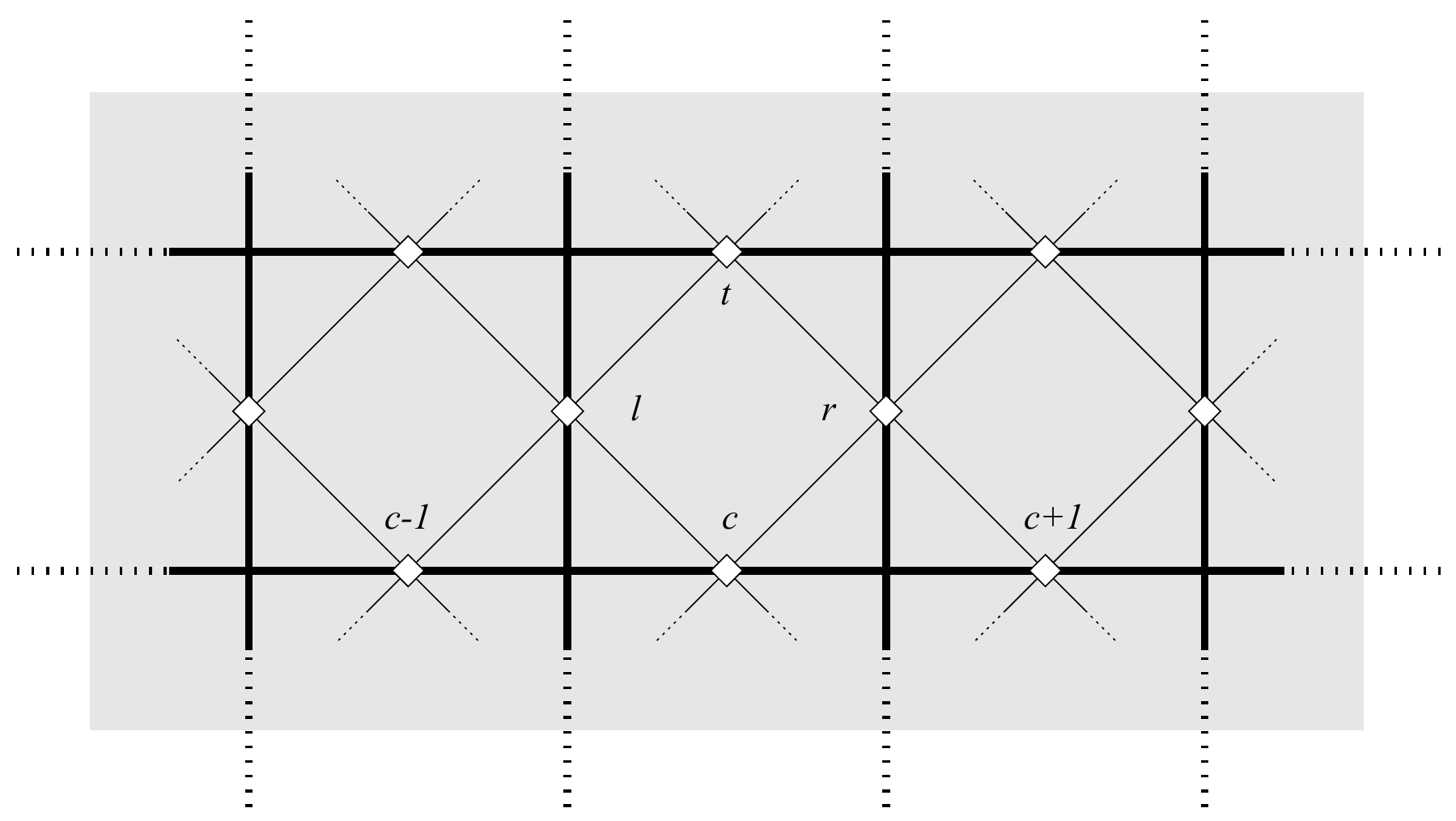}

\includegraphics[width=14cm]{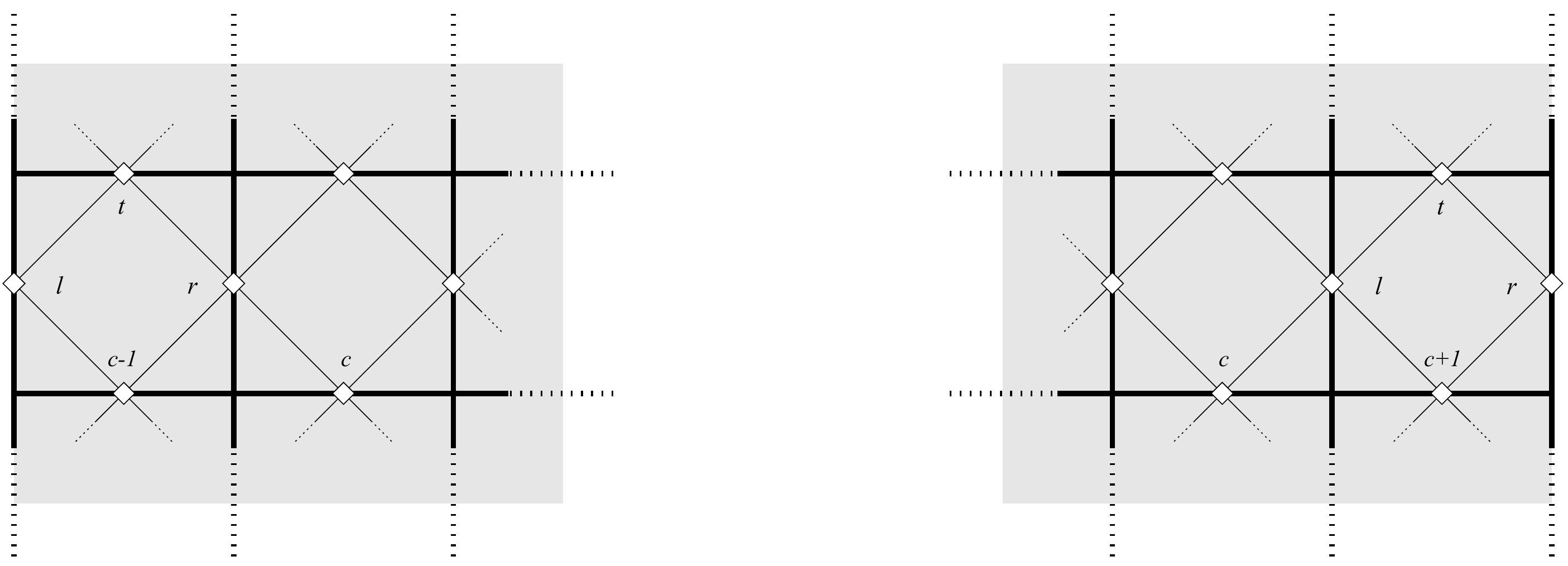}

\caption{\label{fig:Propagations}The values of a massive s-holomorphic function
in the row above can be solved in terms of the values in the row below,
both in the bulk and on the boundary.}

\end{figure}

\begin{defn}
\label{def:massive-shol-prop} Let $\mathbf{I}$ be an interval of
$\mathbb{Z}$ as above. We define the $\beta$-massive s-holomorphic
propagator $P_{\beta}:\left(\mathbb{R}^{2}\right)^{\mathbf{I}^{*}}\to\left(\mathbb{R}^{2}\right)^{\mathbf{I}^{*}}$
by $P_{\beta}f=h\big|_{\mathbf{I}_{1}^{*}}$, where $f$ and $h$
are as in Lemma \ref{lem:massive-shol-prop}
\end{defn}
We can explicitly write down the s-holomorphic propagator in critical
and massive cases. The explicit form will be useful in the next section.
\begin{lem}
\label{lem:s-hol-propagator-formula} Let $\mathbf{I}^{*}=\left\{ a+\frac{1}{2},a+\frac{3}{2},\ldots,b-\frac{1}{2}\right\} $
and denote the left and right extremities by $k_{L}=a+\half$ and
$k_{R}=b-\half$. Set $\lambda:=e^{\ii\pi/4}$. The s-holomorphic
propagator $P$ is given by

\begin{align*}
\begin{cases}
\big(P\, f\big)(k)= & \frac{\lambda^{-3}}{\sqrt{2}}\, f(k-1)+2\, f(k)+\frac{\lambda^{3}}{\sqrt{2}}\, f(k+1)\\
 & +\frac{1}{\sqrt{2}}\,\overline{f(k-1)}-\sqrt{2}\:\overline{f(k)}+\frac{1}{\sqrt{2}}\,\overline{f(k+1)}\qquad\qquad\forall k\in\mathbf{I}^{*}\setminus\left\{ k_{L},k_{R}\right\} \\
\big(P\, f\big)(k_{L})= & \left(1+\frac{1}{\sqrt{2}}\right)f(k_{L})+\frac{\lambda^{3}}{\sqrt{2}}\, f(k_{L}+1)\\
 & +\left(\lambda^{3}+\frac{\lambda^{-3}}{\sqrt{2}}\right)\overline{f(k_{L})}+\frac{1}{\sqrt{2}}\overline{f(k_{L}+1)}\\
\big(P\, f\big)(k_{R})= & \frac{\lambda^{-3}}{\sqrt{2}}\, f(k_{R}-1)+\left(1+\frac{1}{\sqrt{2}}\right)f(k_{R})\\
 & +\frac{1}{\sqrt{2}}\overline{f(k_{R}-1)}+\left(\lambda^{-3}+\frac{\lambda^{3}}{\sqrt{2}}\right)\overline{f(k_{R})}.
\end{cases}
\end{align*}
For $\beta\neq\beta_{c}$, denote $S:=\sinh\left(2\beta\right)$,
$C:=\cosh\left(2\beta\right)$. The massive s-holomorphic propagator
$P_{\beta}$ is given by
\end{lem}
\begin{align*}
\begin{cases}
\big(P_{\beta}\, f\big)(k)= & \frac{-S-\ii}{2S}f(k-1)+\frac{C^{2}}{S}\, f(k)+\frac{-S+\ii}{2S}f(k+1)\\
 & +\frac{C}{2S}\overline{f(k-1)}-C\overline{f(k)}+\frac{C}{2S}\overline{f(k+1)}\qquad\qquad\forall k\in\mathbf{I}^{*}\setminus\left\{ k_{L},k_{R}\right\} \\
\big(P_{\beta}\, f\big)(k_{L})= & \frac{(S+C)C}{2S}\, f(k_{L})+\frac{-S+\ii}{2S}\, f(k_{L}+1)\\
 & +\frac{-(S+C)S+\ii(C-S)}{2S}\,\overline{f(k_{L})}+\frac{C}{2S}\,\overline{f\left(k_{L}+1\right)}\\
\big(P_{\beta}\, f\big)(k_{R})= & \frac{-S-\ii}{2S}\, f(k_{R}-1)+\left(\frac{(S+C)C}{2S}\right)\, f(k_{R})\\
 & +\frac{C}{2S}\,\overline{f(k_{R}-1)}+\left(\frac{-(S+C)S+\ii(C-S)}{2S}\right)\,\overline{f(k_{R})}.
\end{cases}
\end{align*}

\subsection{\label{sub:diag-prop}Spectral splitting of the propagator}
\begin{prop}
\label{prop:form-of-the-eigenvalues}The matrix $P^{\beta}$ is symmetric,
with eigenvalues $\lambda_{\alpha}^{\pm1}$, where $\lambda_{\alpha}>1$
are distinct for $\alpha=1,\ldots,\left|\mathbf{I}\right|^{*}$. \end{prop}
\begin{proof}
Clearly $P_{\beta}$ is invertible: the inverse of $P_{\beta}$ is
the propagation of values of a massive s-holomorphic function downwards.
Notice that exchanging $S$ and $N$ in the massive s-holomorphic
equations (\ref{eq: massive s-holomorphicity relations}) amounts
to replacing $F$ by $\ii\overline{F}$. Denoting by $j:\left(\mathbb{R}^{2}\right)^{\mathbf{I}^{*}}\to\left(\mathbb{R}^{2}\right)^{\mathbf{I}^{*}}$
the (real-linear) involution $f\mapsto\ii\overline{f}$, we deduce
that $\left(P_{\beta}\right)^{-1}=j\circ P_{\beta}\circ j^{-1}$.
We deduce that the spectrum of $P_{\beta}$ is the same as the one
of $\left(P_{\beta}\right)^{-1}$ and hence that the eigenvalues are
of the form $\lambda_{\alpha}^{\pm1}$ with $\lambda_{\alpha}\neq0$
for $\alpha=1,\ldots,\left|\mathbf{I}^{*}\right|$. 

For $\eta\in\mathbb{C}$, observe that the real-linear transpose of
the map $z\mapsto\eta z$ is $z\mapsto\overline{\eta}z$ and that
the map $z\mapsto\eta\overline{z}$ is real-symmetric. From the formula
of Lemma \ref{lem:s-hol-propagator-formula}, we deduce that $P_{\beta}$
is symmetric. To show that $P_{\beta}$ is positive definite, it is
enough to show this at the critical temperature $\beta=\beta_{c}$,
since the eigenvalues of $P_{\beta}$ are continuous in $\beta$ and
cannot be zero. We can write the propagator $P_{\beta}$ as $BA$,
where $A:\left(\mathbb{R}^{2}\right)^{\mathbf{I}^{*}}\to\left(\mathbb{R}^{2}\right)^{\mathbf{I}}$
is the propagation of $f:\mathbf{I}_{0}^{*}\to\mathbb{C}$ to $\mathbf{I}_{\frac{1}{2}}$
(see Definition \ref{def:massive-s-hol}) and $B:\left(\mathbb{R}^{2}\right)^{\mathbf{I}}\to\left(\mathbb{R}^{2}\right)^{\mathbf{I}^{*}}$
is the propagation of $g:\mathbf{I}_{\frac{1}{2}}\to\mathbb{C}$ to
$\mathbf{I}_{1}^{*}$. At $\beta=\beta_{c}$, we have that $B=A^{\top}$
and hence $P=A^{\top}A$ is positive definite. 

Let us now show that $1$ cannot be an eigenvalue. Suppose $f:\mathbf{I}_{0}^{*}\to\mathbb{C}$
is such that $P_{\beta}f=f$; we want to show that $f=0$. Let $h:\mathbf{I}_{0}^{*}\cup\mathbf{I}_{\frac{1}{2}}\cup\mathbf{I}_{1}^{*}\to\mathbb{C}$
denote the massive s-holomorphic extension of $f$.

At $\beta=\beta_{c}$ (i.e. when $P_{\beta}=P$), there is a particularly
simple argument. Since $h$ satisfies the discrete Cauchy-Riemann
equations, i.e. for any $x\in\mathbf{I}_{\frac{1}{2}}^{*}$, we have
\[
h\left(x+\frac{1}{2}\right)-h\left(x-\frac{1}{2}\right)=\frac{1}{\ii}\left(h\left(x+\frac{\ii}{2}\right)-h\left(x-\frac{\ii}{2}\right)\right)=\frac{1}{\ii}\left(f\left(x\right)-f\left(x\right)\right)=0.
\]
Hence $h$ must be constant on $\mathbf{I}_{\frac{1}{2}}$ and the
Riemann boundary conditions easily imply that $h=0$. In turn, this
implies that $f=0$, by the s-holomorphicity equations.

For general $\beta$, writing $\mathbf{I}_{\frac{1}{2}}=\left\{ x_{L},x_{L}+1,\ldots,x_{R}\right\} $,
we can deduce (from an explicit computation) that $h\left(x+1\right)$
and $h\left(x\right)$ (for $x\in\left\{ x_{L},\ldots,x_{R}-1\right\} $)
satisfy a linear relation: 
\begin{equation}
h\left(x+1\right)+\ii\mathcal{B}\overline{h\left(x+1\right)}=h\left(x\right)-\ii\mathcal{B}\overline{h\left(x\right)},\label{eq:recursion-h-1-eigenvect}
\end{equation}
for $\mathcal{B}=\sqrt{2}\Im\mathfrak{m}\left(\nu\right)$, where
$\nu=e^{\ii3\pi/4}\frac{\alpha+\ii}{\alpha-\ii}$ and $\alpha=e^{-2\beta}$
as above.

The Riemann boundary condition on the left extremity imposes that
$h\left(x_{L}\right)=\mathcal{C}e^{-\ii\pi/4}$ for some $\mathcal{C}\in\mathbb{R}$,
and the above equation (\ref{eq:recursion-h-1-eigenvect}) yields
$h\left(x\right)=\mathcal{C}e^{-\ii\pi/4}\mathcal{A}^{x-x_{L}}$ where
$\mathcal{A}=\frac{1+\mathcal{B}}{1-\mathcal{B}}\neq0$. But the Riemann
boundary condition on the right extremity, $h\left(x_{R}\right)\parallel e^{\ii\pi/4}$,
then requires that $\mathcal{C}\mathcal{A}^{x_{R}-x_{L}}=0$ and hence
$h=0$ everywhere by the massive s-holomorphic equations.

Finally we show that the eigenvalues are distinct. Suppose $\Lambda>0$
is an eigenvalue of $P_{\beta}$, and let $f_{\Lambda}\in(\bR^{2}){}^{\mathbf{I}^{*}}$
is an eigenvector, and let $h_{\Lambda}$ be the massive s-holomorphic
extension of $f_{\Lambda}$ to $\mathbf{I}_{0}^{*}\cup\mathbf{I}_{\frac{1}{2}}\cup\mathbf{I}_{1}^{*}$.
The massive s-holomorphicity equations can be solved to obtain a recursion
relation $h_{\Lambda}(x+1)=\eta h_{\Lambda}(x)+\eta'\overline{h_{\Lambda}(x)}$
with some explicit $\eta,\eta'\in\bC$. This shows that the eigenspace
is one-dimensional.
\end{proof}

\section{\label{sec:tm-cliff-alg-induc-rot}Transfer Matrix, Clifford Algebra
and Induced Rotation}

In this section we review fundamental algebraic structures underlying
the transfer matrix formalism introduced in \cite{kaufman}. See \cite{palmer}
for a recent exposition with more details.

\subsection{\label{sub:symm-tm}\label{sub:clifford-generators}\label{sub:clifford-algebra}Transfer
matrix and Clifford algebra}

In the introduction, Section \ref{sub:transfer-matrix-approach},
we defined the Ising transfer matrix $V:\mathcal{S}\to\mathcal{S}$
with fixed boundary conditions at the two extremities of the row $\mathbf{I}=\set{a,a+1,\ldots,b-1,b}$
as the product

\begin{align*}
V=\; & (V^{\mathrm{h}})^{\half}\, V^{\mathrm{v}}\,(V^{\mathrm{h}})^{\half},
\end{align*}
where the matrix elements in the basis $(\mathbf{e}_{\sigma})$ indexed
by spin configurations in a row, $\sigma\in\set{\pm1}^{\mathbf{I}}$,
are given by
\[
V_{\sigma\rho}^{\mathrm{v}}=\begin{cases}
\exp\left(\beta\sum_{i=a}^{b}\sigma_{i}\rho_{i}\right)\quad & \mbox{if }\sigma_{a}=\rho_{a}\text{ and }\sigma_{b}=\rho_{b}\\
0\qquad & \mbox{otherwise}
\end{cases}
\]
and 
\[
\left(V_{\sigma\rho}^{\mathrm{h}}\right)^{\half}=\begin{cases}
\exp\left(\frac{\beta}{2}\sum_{i=a}^{b-1}\sigma_{i}\sigma_{i+1}\right)\quad & \mbox{if }\sigma\equiv\rho\\
0\qquad & \mbox{otherwise}.
\end{cases}
\]

There is a two-fold degeneracy in the spectrum of the transfer matrix:
the global spin flip $\sigma\mapsto-\sigma$ commutes with both $V^{\mathrm{v}}$
and $\left(V^{\mathrm{h}}\right)^{\frac{1}{2}}$. To disregard the
corresponding multiplicity of eigenvalues, we restrict our attention
to the subspace
\begin{align*}
\sS_{+}=\; & \spn\set{\mathbf{e}_{\sigma}\,\big|\,\sigma_{b}=+1},
\end{align*}
spanned by spin configurations that have a plus spin on the right
extremity of the row. This subspace is invariant for both $V^{\mathrm{v}}$
and $\left(V^{\mathrm{h}}\right)^{\frac{1}{2}}$. Note that the dimension
is given by 
\begin{align*}
\dmn\left(\sS_{+}\right)=\; & 2^{|\mathbf{I}|-1}=2^{|\mathbf{I}^{*}|}=2^{b-a}.
\end{align*}

In Section \ref{sub:intro-induced-rotation}, we defined also the
operators $p_{k}$ and $q_{k}$ on $\sS$, for $k\in\mathbf{I}^{*}=\set{a+\half,\, a+\frac{3}{2},\,\ldots,\, b-\frac{3}{2},\, b-\half}$,
by 
\begin{align*}
 & \begin{array}{c}
p_{k}\left(\mathbf{e}_{\sigma}\right)=\;\sigma_{k+\frac{1}{2}}\mathbf{e}_{\tau}\\
q_{k}\left(\mathbf{e}_{\sigma}\right)=\ii\,\sigma_{k-\frac{1}{2}}\mathbf{e}_{\tau}
\end{array}\qquad\text{, where} & \tau_{x}=\; & \begin{cases}
\phantom{-}\sigma_{x\quad} & \text{for }x>k\\
-\sigma_{x}\quad & \text{for }x<k.
\end{cases}
\end{align*}
Again, the subspace $\sS_{+}\subset\mathcal{S}$ is invariant for
all $p_{k},q_{k}$. We denote $\mathcal{W}=\spn\set{p_{k},q_{k}\,\big|\, k\in\mathbf{I}^{*}}\subset{\rm End}(\sS)$.

It is easy to check that $p_{k},q_{k}$ satisfy the relations 
\begin{eqnarray*}
p_{k}p_{\ell}+p_{\ell}p_{k} & = & 2\,\delta_{k,\ell}\,\id_{\sS}\\
q_{k}q_{\ell}+q_{\ell}q_{k} & = & 2\,\delta_{k,\ell}\,\id_{\sS}\\
p_{k}q_{\ell}+q_{\ell}p_{k} & = & 0,
\end{eqnarray*}
i.e. that they form a Clifford algebra representation on $\mathcal{S}_{+}$
and on $\mathcal{S}$. This representation is faithful, so we think
of the Clifford algebra $\Cliff$ simply as the algebra of linear
operators $\sS\rightarrow\sS$ generated by $\sW$. 

Consider the symmetric bilinear form $(\cdot,\cdot)$ on $\sW$ given
by $(p_{k},p_{l})=2\delta_{k,l}$, $(q_{k},q_{l})=2\delta_{k,l}$,
$(p_{k},q_{l})=0$. Then $\Cliff$ is the algebra with set of generators
$\sW$ and relations $uv+vu=(u,v)\,\mathbf{1}$, for $u,v\in\sW$.
The dimensions of the set of Clifford generators and the Clifford
algebra are
\begin{align*}
\dmn(\sW)=\; & 2|\mathbf{I}^{*}|=2(b-a), & \dmn(\Cliff)=\; & 2^{\dmn(\sW)}=2^{2|\mathbf{I}^{*}|}=2^{2(b-a)}.
\end{align*}

The transfer matrix can be written in terms of exponentials of quadratic
expressions in the Clifford algebra generators as follows.
\begin{prop}
We have
\begin{eqnarray*}
\left(V^{\mathrm{h}}\right)^{\half} & = & \exp\left(\ii\frac{\beta}{2}\sum_{k\in\mathbf{I}^{*}}q_{k}p_{k}\right)\\
V^{\mathrm{v}} & = & e^{2\beta}\left(2S\right)^{\frac{\left|\mathbf{I}\right|}{2}-1}\;\exp\left(\ii\beta^{*}\sum_{j\in\mathbf{I}\setminus\bdry\mathbf{I}}p_{j-\frac{1}{2}}q_{j+\frac{1}{2}}\right),
\end{eqnarray*}
where $\beta^{*}$ is the dual inverse temperature given by $\tanh\left(\beta^{*}\right)=e^{-2\beta}$
and $S=\sinh\left(2\beta\right)$.\end{prop}
\begin{proof}
Note that the operator $\ii\, q_{k}p_{k}$ has the following diagonal
action in the basis $(\mathbf{e}_{\sigma})$
\begin{align*}
\ii\, q_{k}p_{k}\;\mathbf{e}_{\sigma}=\; & \sigma_{k+\half}\sigma_{k-\half}\;\mathbf{e}_{\sigma},
\end{align*}
so the first asserted result $e^{\ii\frac{\beta}{2}\sum_{k}q_{k}p_{k}}=(V^{\mathrm{h}})^{\half}$
follows immediately. The operator $\ii\, p_{j-\half}q_{j+\half}$
inverts the value of the spin at $j$, 
\begin{align*}
\ii\, p_{j-\half}q_{j+\half}\;\mathbf{e}_{\sigma}=\; & \mathbf{e}_{\sigma'},\qquad\text{where }\sigma_{x}'=\begin{cases}
\phantom{-}\sigma_{x}\qquad & \text{for }x\neq j\\
-\sigma_{j}\qquad & \text{for }x=j,
\end{cases}
\end{align*}
so we have
\begin{align*}
\exp\left(\ii\beta^{*}p_{j-\half}q_{j+\half}\right)\,\mathbf{e}_{\sigma}=\; & \cosh(\beta^{*})\,\mathbf{e}_{\sigma}+\sinh(\beta^{*})\,\mathbf{e}_{\sigma'}.
\end{align*}
Taking $\beta^{*}$ such that $\tanh\left(\beta^{*}\right)=e^{-2\beta}$
and computing the product of these commuting operators at different
$j$ we get
\begin{align*}
\exp\left(\ii\beta^{*}\sum_{j=a+1}^{b-1}p_{j-\half}q_{j+\half}\right)\,\mathbf{e}_{\sigma}=\; & \cosh(\beta^{*})^{b-a-1}\sum_{\substack{\rho\in\set{\pm1}^{\mathbf{I}}\\
\rho_{a}=\sigma_{a},\;\rho_{b}=\sigma_{b}
}
}\tanh(\beta^{*})^{\#\{x|\sigma_{x}\neq\rho_{x}\}}\,\mathbf{e}_{\rho}\\
=\; & e^{-2\beta}\left(2S\right)^{1-\frac{\left|\mathbf{I}\right|}{2}}\sum_{\substack{\rho\in\set{\pm1}^{\mathbf{I}}\\
\rho_{a}=\sigma_{a},\;\rho_{b}=\sigma_{b}
}
}\exp\left(\sum_{i=a+1}^{b-1}\sigma_{i}\rho_{i}\right)\,\mathbf{e}_{\rho},
\end{align*}
which is the second asserted result. 
\end{proof}

\subsection{\label{sub:induced-rotation}Induced rotation}

Since the constituents of the transfer matrix are exponentials of
second order polynomials in the Clifford generators, conjugation by
the transfer matrix stabilizes the set of Clifford generators. In
the formulas below, we use the notation
\begin{align*}
s=\; & \sinh(\beta)\qquad & s^{*}=\; & \sinh(\beta^{*})\\
c=\; & \cosh(\beta)\qquad & c^{*}=\; & \cosh(\beta^{*})\\
S=\; & \sinh(2\beta)=\frac{1}{\sinh(2\beta^{*})}\qquad & C=\; & \cosh(2\beta)=\frac{\cosh(2\beta^{*})}{\sinh(2\beta^{*})}.
\end{align*}
The following lemma is a result of straightforward calculations, which
can be found e.g. in \cite{palmer}.
\begin{lem*}
\label{lem:massive-s-hol-prop-induced-rotation} Conjugation by $V^{\mathrm{h}}$
is given by the following formulas on Clifford generators $p_{k},q_{k}$
($k\in\mathbf{I}^{*}$) 
\begin{align*}
\left(V^{\mathrm{h}}\right)^{-\half}\circ p_{k}\circ\left(V^{\mathrm{h}}\right)^{\half}=\; & c\, p_{k}-\ii s\, q_{k}\\
\left(V^{\mathrm{h}}\right)^{-\half}\circ q_{k}\circ\left(V^{\mathrm{h}}\right)^{\half}=\; & \ii s\, p_{k}+c\, q_{k}.
\end{align*}
Let $k_{L}=a+\half$ and $k_{R}=b-\half$ be the leftmost and rightmost
points of $\mathbf{I}^{*}$. Conjugation by $V^{\mathrm{v}}$ is given
by
\begin{align*}
\left(V^{{\rm v}}\right)^{-1}\circ p_{k}\circ V^{{\rm v}}=\; & \frac{C}{S}\, p_{k}+\frac{\ii}{S}\, q_{k+1}\qquad\text{for }k\neq k_{R}\\
\left(V^{{\rm v}}\right)^{-1}\circ q_{k}\circ V^{{\rm v}}=\; & \frac{-\ii}{S}\, p_{k-1}+\frac{C}{S}\, q_{k}\qquad\text{for }k\neq k_{L}
\end{align*}
and on the remaining generators by 
\begin{align*}
\left(V^{{\rm v}}\right)^{-1}\circ p_{k_{R}}\circ V^{{\rm v}}=\; & p_{k_{R}}\qquad & \left(V^{{\rm v}}\right)^{-1}\circ q_{k_{L}}\circ V^{{\rm v}}=\; & q_{k_{L}}.
\end{align*}

\end{lem*}
We see that $\sW\subset\End(\sS)$ is an invariant subspace for the
conjugation by the transfer matrix $V$. The conjugation is called
the \emph{induced rotation} of $V$, and denoted by
\begin{align*}
T_{V}\colon\; & \sW\rightarrow\sW\qquad & T_{V}(w)=\; & V\circ w\circ V^{-1}.
\end{align*}
Note that the induced rotation $T_{V}$ preserves the bilinear form,
$(T_{V}u,\, T_{V}v)=(u,v)$ for all $u,v\in\sW$.

We will next show that the induced rotation is, up to a change of
basis, the complexification of the row-to-row propagation $P_{\beta}$
of massive s-holomorphic functions satisfying the Riemann boundary
condition. To facilitate the calculations, we introduce two symmetry
operations on the set of Clifford algebra generators. We define a
(complex) linear isomorphism $R\colon\sW\rightarrow\sW$ and a (complex)
conjugate-linear isomorphism $J\colon\sW\rightarrow\sW$ by the formulas
\begin{align*}
R(p_{k})=\; & \ii\, q_{a+b-k}\quad & R(q_{k})=\; & -\ii\, p_{a+b-k}\quad & \text{(extended linearly)}\\
J(p_{k})=\; & \ii\, p_{k}\quad & J(q_{k})=\; & -\ii\, q_{k}\quad & \text{(extended conjugate-linearly)}.
\end{align*}
We will moreover use for $\sW$ the basis
\begin{align*}
\psi_{k}=\; & \frac{\ii}{\sqrt{2}}(p_{k}+q_{k}), & \bar{\psi}_{k}=\; & \frac{1}{\sqrt{2}}(p_{k}-q_{k}), & k\in\mathbf{I}^{*}.
\end{align*}

\begin{lem}
\label{lem:symmetries-on-Clifford-generators} The maps $R$ and $J$
commute with $T_{V}$, i.e. we have
\begin{align*}
T_{V}\circ R=\; & R\circ T_{V}, & T_{V}\circ J=\; & J\circ T_{V}.
\end{align*}
For all $k\in\mathbf{I}^{*}$ we have 
\begin{align*}
R(\psi_{k})=\; & \bar{\psi}_{b+a-k}, & R(\bar{\psi}_{k})=\; & \psi_{b+a-k}\\
J(\psi_{k})=\; & \bar{\psi}_{k}, & J(\bar{\psi}_{k})=\; & \psi_{k}.
\end{align*}
\end{lem}
\begin{proof}
By the explicit expressions of Lemma \ref{lem:massive-s-hol-prop-induced-rotation}
one easily verifies that $R$ and $J$ commute with the conjugation
by both $\left(V^{{\rm h}}\right)^{\half}$ and $V^{{\rm v}}$.\end{proof}
\begin{thm}
\label{thm:massive-s-hol-prop-induced-rotation} The induced rotation
$T_{V}\colon\sW\rightarrow\sW$ is up to a change of basis equal to
the complexification of $P_{\beta}$, i.e. there exists a linear isomorphism
$\varrho:\left(\mathbb{C}^{2}\right)^{\mathbf{I}^{*}}\to\mathcal{W}$
such that $T_{V}=\varrho\circ P^{\mathbb{C}}\circ\varrho^{-1}$.\end{thm}
\begin{proof}
Consider the action of the induced rotation in the basis $(\psi_{k},\bar{\psi}_{k})_{k\in\mathbf{I}^{*}}$.
With the formulas of Lemma \ref{lem:massive-s-hol-prop-induced-rotation},
it is straighforward to compute that for $k\in\mathbf{I}^{*}\setminus(\bdry\mathbf{I}^{*})$
\begin{align*}
T_{V}^{-1}(\psi_{k})=\; & \frac{C^{2}}{S}\,\psi_{k}+\left(-\half-\frac{\ii}{2S}\right)\,\psi_{k-1}+\left(-\half+\frac{\ii}{2S}\right)\,\psi_{k+1}\\
 & -C\,\bar{\psi}_{k}+\frac{1}{2}\frac{C}{S}\,\bar{\psi}_{k-1}+\frac{1}{2}\frac{C}{S}\,\bar{\psi}_{k+1}.
\end{align*}
To get a formula for $T_{V}^{-1}(\bar{\psi}_{k})$, $k\in\mathbf{I}^{*}\setminus(\bdry\mathbf{I}^{*})$,
apply the map $J$ on this and use Lemma \ref{lem:symmetries-on-Clifford-generators}.
We still need formulas for the two extremities, $k\in\bdry\mathbf{I}^{*}=\set{a+\half,b-\half}$.
On the left extremity, at $k_{L}=a+\half$, a straightforward calculation
yields
\begin{align*}
T_{V}^{-1}(\psi_{k_{L}})=\; & \frac{C(S+C)}{2S}\,\psi_{k_{L}}+\frac{(1+\ii S)}{2S}\ii\,\psi_{k_{L}+1}\\
 & +\frac{-S(C+S)+\ii(C-S)}{2S}\,\bar{\psi}_{k_{L}}+\frac{C}{2S}\,\bar{\psi}_{k_{L}+1}.
\end{align*}
To get a formula for $T_{V}^{-1}(\bar{\psi}_{k_{L}})$, apply the
map $J$ on this. To get the formula for $T_{V}^{-1}(\bar{\psi}_{k_{R}})$,
where $k_{R}=b-\half$, apply the map $R$. To get a formula for $T_{V}^{-1}(\psi_{k_{R}})$,
apply the composition $J\circ R$.

Since the coefficients in the formulas for $T_{V}^{-1}(\psi_{k})$
coincide with the coefficients in the formulas for $(P_{\beta}f)(k)$
in Section \ref{sub:s-hol-prop}, and the coefficients in the formulas
for $T_{V}^{-1}(\bar{\psi}_{k})$ are the complex conjugates of the
corresponding ones, we get that the complexification of $P_{\beta}$
agrees with $T_{V}^{-1}$ up to a change of basis. This finishes the
proof, because $P_{\beta}$ and its inverse $P_{\beta}^{-1}$ are
conjugates by Proposition \ref{prop:form-of-the-eigenvalues}.
\end{proof}

\subsection{\label{sub:fock-representations}Fock representations}

Finite dimensional irreducible representations of the Clifford algebra
$\Cliff$ are Fock representations, defined below. To define a Fock
representation, one first chooses a way to split the set $\sW$ of
Clifford algebra generators to creation and annihilation operators.
Let $(\cdot,\cdot)$ denote the bilinear form on $\sW$ defined in
Section \ref{sub:clifford-algebra}. A \emph{polarization} (an \emph{isotropic
splitting}) is a choice of two complementary subspaces $\sW_{\cre}$
(creation operators) and $\sW_{\ann}$ (annihilation operators) of
the set of Clifford algebra generators,
\begin{align*}
\sW=\; & \sW_{\cre}\oplus\sW_{\ann},
\end{align*}
such that
\begin{align*}
(w_{\cre},w_{\cre}')=\; & 0\qquad\text{for all }w_{\cre},w_{\cre}'\in\sW_{\cre}\\
(w_{\ann},w_{\ann}')=\; & 0\qquad\text{for all }w_{\ann},w_{\ann}'\in\sW_{\ann}.
\end{align*}
Note that due to the nondegeneracy of the bilinear form $(\cdot,\cdot)$,
the two subspaces $\sW_{\cre}$ and $\sW_{\ann}$ are naturally dual
to each other, and in particular
\begin{align*}
\dmn(\sW_{\cre})=\; & \dmn(\sW_{\ann})=\half\dmn(\sW)=|\mathbf{I}^{*}|=b-a
\end{align*}
is the number of linearly independent creation operators.

As a vector space, the Fock representation corresponding to the polarization
$\sW=\sW_{\cre}\oplus\sW_{\ann}$ is the exterior algebra of $\sW_{\cre}$,
\begin{align*}
\bigwedge\sW_{\cre}=\; & \bigoplus_{n=0}^{|\mathbf{I}^{*}|}\left(\wedge^{n}\sW_{\cre}\right).
\end{align*}
To define the representation of the Clifford algebra on this vector
space, let $(a_{\alpha}^{\dagger})_{\alpha=1}^{|\mathbf{I}^{*}|}$
be a basis of $\sW_{\cre}$ and $(a_{\alpha})_{\alpha=1}^{|\mathbf{I}^{*}|}$
the dual basis of $\sW_{\ann}$, i.e. $(a_{\alpha}^{\dagger},a_{\beta})=\delta_{\alpha,\beta}$.
The action of the Clifford algebra on the Fock space $\bigwedge\sW_{\cre}$
is given by the linear extension of the formulas
\begin{align*}
a_{\alpha}^{\dagger}.(a_{\beta_{1}}^{\dagger}\wedge a_{\beta_{2}}^{\dagger}\wedge\cdots\wedge a_{\beta_{n}}^{\dagger})=\; & a_{\alpha}^{\dagger}\wedge a_{\beta_{1}}^{\dagger}\wedge a_{\beta_{2}}^{\dagger}\wedge\cdots\wedge a_{\beta_{n}}^{\dagger}\\
a_{\alpha}.(a_{\beta_{1}}^{\dagger}\wedge a_{\beta_{2}}^{\dagger}\wedge\cdots\wedge a_{\beta_{n}}^{\dagger})=\; & \sum_{j=1}^{n}(-1)^{j-1}\delta_{\alpha,\beta_{j}}\; a_{\beta_{1}}^{\dagger}\wedge\cdots\wedge a_{\beta_{j-1}}^{\dagger}\wedge a_{\beta_{j+1}}^{\dagger}\wedge\cdots\wedge a_{\beta_{n}}^{\dagger}.
\end{align*}
The vector $1\in\mathbb{C}\isom\wedge^{0}\sW_{\cre}\subset\bigwedge\sW_{\cre}$
is called the vacuum of the Fock representation: it is annihilated
by all of $\sW_{\ann}$. Irreducible representations are characterized
by such vacuum vectors as follows.
\begin{lem}
\label{lem:Fock representation facts} Suppose that $\sW=\sW_{\cre}\oplus\sW_{\ann}$
is a polarization. Any irreducible representation of $\Cliff$ is
isomorphic to the Fock representation $\bigwedge\sW_{\cre}$. If a
representation $\mathcal{V}$ of $\Cliff$ contains a non-zero vector
$v_{\mathrm{vac}}\in\mathcal{V}$ satisfying $\sW_{\ann}v_{\mathrm{vac}}=0$,
then the Fock space $\bigwedge\sW_{\cre}$ embeds in $\mathcal{V}$
by the mapping
\begin{align*}
a_{\beta_{1}}^{\dagger}\wedge a_{\beta_{2}}^{\dagger}\wedge\cdots\wedge a_{\beta_{n}}^{\dagger}\mapsto\; & (a_{\beta_{1}}^{\dagger}a_{\beta_{2}}^{\dagger}\cdots a_{\beta_{n}}^{\dagger}).v_{\mathrm{vac}}.
\end{align*}
\end{lem}
\begin{proof}
Consider a representation $\mathcal{V}$ of $\Cliff$. Choose a non-zero
$v^{(0)}\in\mathcal{V}$. Define recursively $v^{(\alpha)}$, for
$\alpha=1,2,\ldots|\mathbf{I}^{*}|$, to be $a_{\alpha}v^{(\alpha-1)}$
if this is non-vanishing and $v^{(\alpha-1)}$ otherwise. Then $v_{\mathrm{vac}}=v^{(|\mathbf{I}^{*}|)}$
is non-zero and $\sW_{\ann}v_{\mathrm{vac}}=0$, i.e. $v_{\mathrm{vac}}$
is a vacuum vector. This argument shows in particular that any non-zero
subrepresentation of the Fock representation $\bigwedge\sW_{\cre}$
contains the vacuum vector $1\in\wedge^{0}\sW_{\cre}$, and thus the
Fock representation is irreducible. The mapping $\bigwedge\sW_{\cre}\rightarrow\mathcal{V}$
given in the statement defines a non-zero intertwining map of Clifford
algebra representations, and by irreducibility of the Fock representation,
this is an embedding. If $\mathcal{V}$ is irreducible the embedding
must be surjective, and thus an isomorphism.
\end{proof}
The fact that the Fock representation is the only isomorphism type
of irreducible representations of the Clifford algebra would follow
already from the irreducibility of the Fock representation and the
observation that $\left(\dmn(\bigwedge\sW_{\cre})\right)^{2}=\dmn(\Cliff)$,
by the (Artin\textendash{})Wedderburn structure theorem.

A standard tool for performing calculations in the Fock representation
is the following. We recall that the Pfaffian ${\rm Pf}(A)$ of an
antisymmetric matrix $A\in\bC^{n\times n}$ is zero if $n$ is odd,
and if $n=2m$ is even, then it is given by
\begin{align*}
{\rm Pf}(A)=\; & \frac{1}{2^{m}\, m!}\sum_{\pi\in S_{2m}}\sgn(\pi)\,\prod_{s=1}^{m}A_{\pi(2s-1),\pi(2s)}.
\end{align*}

\begin{lem}[Fermionic Wick's formula]
\label{lem:Wick formula} Let $\sW=\sW_{\cre}\oplus\sW_{\ann}$ be
a polarization, and consider the Fock representation $\bigwedge\sW_{\cre}$.
Let $v_{{\rm vac}}=1\in\wedge^{0}\sW_{\cre}\subset\bigwedge\sW_{\cre}$
be the vacuum and $v_{{\rm vac}}^{*}\in\left(\wedge^{0}\sW_{\cre}\right)^{*}\subset\left(\bigwedge\sW_{\cre}\right)^{*}$
be the dual vacuum normalized by $\langle v_{{\rm vac}}^{*},v_{{\rm vac}}\rangle=1$.
Then for any $\phi_{1},\ldots,\phi_{n}\in\sW$ we have
\begin{align*}
\langle v_{{\rm vac}}^{*},\;\phi_{1}\cdots\phi_{n}v_{{\rm vac}}\rangle=\; & \mathrm{Pf}\left(\Big[\langle v_{{\rm vac}}^{*},\,\phi_{i}\phi_{j}v_{{\rm vac}}\rangle\Big]_{i,j=1}^{n}\right).
\end{align*}
\end{lem}
\begin{proof}
Write the elements $\phi_{i}\in\sW$ as sums of creation and annihilation
operators, and then anticommute the annihilation operators to the
right and the creation operators to the left.
\end{proof}

\subsection{A simple polarization for low temperature expansions}

The following lemma gives one of the simplest possible polarizations.
\begin{lem}
\label{lem:simple-polarization} The following formulas define a polarization
\begin{align*}
\sW_{\cre}^{(+)}=\; & \spn\set{p_{k}-\ii q_{k}\;\big|\; k\in\mathbf{I}^{*}}\\
\sW_{\ann}^{(+)}=\; & \spn\set{p_{k}+\ii q_{k}\;\big|\; k\in\mathbf{I}^{*}}.
\end{align*}
\end{lem}
\begin{proof}
The vectors $p_{k}+\ii q_{k}$, $p_{k}-\ii q_{k}$ for $k\in\mathbf{I}^{*}$
form a basis of $\sW$, and we have $(p_{k}\pm\ii q_{k},\, p_{l}\pm\ii q_{l})=2\delta_{k,l}+0+0+(\pm\ii)^{2}\,2\delta_{k,l}=0$. 
\end{proof}
Recall that in the state space $\sS$ of the transfer matrix formalism
we have the vectors corresponding to the constant spin configurations
in a row,
\begin{align*}
\mathbf{e}_{(+)}\in\; & \sS_{+} & (+)=\; & (+1,+1,\ldots,+1)\in\set{\pm1}^{\mathbf{I}}\\
\mathbf{e}_{(-)}\in\; & \sS & (-)=\; & (-1,-1,\ldots,-1)\in\set{\pm1}^{\mathbf{I}}.
\end{align*}
Directly from the defining formulas of the operators $p_{k},q_{k}$,
one sees that the vectors $\mathbf{e}_{(+)},\mathbf{e}_{(-)}\in\sS$
satisfy $(p_{k}+\ii q_{k})\mathbf{e}_{(+)}=0$ and $(p_{k}+\ii q_{k})\mathbf{e}_{(-)}=0$
for all $k\in\mathbf{I}^{*}$.
\begin{cor}
\label{cor:simple-polarization} As a representation of the Clifford
algebra, $\sS_{+}$ is isomorphic to the Fock space $\bigwedge\sW_{\cre}^{(+)}$,
with vacuum vector $v_{{\rm vac}}^{(+)}=\mathbf{e}_{(+)}$, and $\sS$
is isomorphic to the direct sum of two copies of this Fock space.
\end{cor}
We emphasize that the polarization of this subsection is not the physical
one, but by Lemma \ref{lem:Fock representation facts}, the isomorphism
type of the Fock representation doesn't depend on the polarization,
so the state space of the transfer matrix formalism is in fact a Fock
representation for any polarization. The polarization is, however,
the zero temperature limit ($\beta\nearrow\infty$) of the physical
polarizations of the next section, and it is very closely related
to the low temperature graphical expansions of correlation functions
and observables considered in Sections \ref{sub:winding-observables-and-low-T-expansions},
\ref{sub:Pfaffian formula} and \ref{sub:Multipoint-parafermionic-observables}.
In particular, a slight modification of this simple polarization together
with the fermionic Wick's formula will be used for the proof of Pfaffian
formulas for fermion operator multi-point correlation functions and
multi-point parafermionic observables. The modified polarization is
the following.
\begin{lem}
\label{lem:simple-polarization-general-N} Let $N\in\bN$. The following
formulas define a polarization 
\begin{align*}
\sW_{\cre}^{(+);N}=\; & \spn\set{V^{-N}(p_{k}-\ii q_{k})V^{N}\;\big|\; k\in\mathbf{I}^{*}}\\
\sW_{\ann}^{(+)}=\; & \spn\set{p_{k}+\ii q_{k}\;\big|\; k\in\mathbf{I}^{*}}
\end{align*}
for all $\beta$ except possibly for isolated values. The space $\sS_{+}$
is isomorphic to a Fock representation $\bigwedge\left(\sW_{\cre}^{(+);N}\right)$,
with vacuum vector $v_{{\rm vac}}^{(+);N}=\mathbf{e}_{(+)}\in\sS_{+}$
and dual vacuum vector $(v_{{\rm vac}}^{(+);N})^{*}=\frac{1}{\mathbf{e}_{(+)}^{\top}V^{N}\mathbf{e}_{(+)}}\times\mathbf{e}_{(+)}^{\top}V^{N}$.\end{lem}
\begin{proof}
The special case $N=0$ was treated above. Since we have $(p_{k}\pm\ii q_{k},\, p_{l}\pm\ii q_{l})=0$
and the bilinear form is invariant under $T_{V}$, $(T_{V}^{-N}u,\, T_{V}^{-N}v)=(u,v)$,
it follows that also for general $N$ the choice of subspaces is a
polarization if the two subspaces span the whole space $\sW$, that
is if the vectors $T_{V}^{-N}(p_{k}-\ii q_{k})$ and $p_{k}+\ii q_{k}$
for $k\in\mathbf{I}^{*}$ form a basis of $\sW$. It suffices to show
that the matrix $\left[(p_{k}+\ii q_{k},\; T_{V}^{-N}(p_{l}-\ii q_{l}))\right]_{k,l\in\mathbf{I}^{*}}$
of the bilinear form is non-degenerate. The non-degeneracy is evident
in the limit $\beta\nearrow\infty$, since $e^{-2\beta}T_{V}^{-1}(p_{k}-\ii q_{k})=p_{k}-\ii q_{k}+O(e^{-\beta})$
and $e^{-2\beta}T_{V}^{-1}(p_{k}+\ii q_{k})=O(e^{-\beta})$ by the
formulas of Lemma \ref{lem:massive-s-hol-prop-induced-rotation}.
The determinant $\det\left(\left[(p_{k}+\ii q_{k},\; T_{V}^{-N}(p_{l}-\ii q_{l}))\right]_{k,l\in\mathbf{I}^{*}}\right)$
is analytic in $\beta$, so its (possible) zeroes can't have accumulation
points. We conclude that $\sW=\sW_{\cre}^{(+);N}\oplus\sW_{\ann}^{(+)}$
is a polarization except possibly for isolated values of $\beta$.

The same calcuation as before shows that $\mathbf{e}_{(+)}$ is a
vacuum vector of the Fock representation $\sS_{+}\isom\bigwedge W_{\cre}^{(+);N}$,
and similarly from the calculation $\mathbf{e}_{(+)}^{\top}(p_{k}-\ii q_{k})=0$
we get that the dual vacuum of $\sS_{+}\isom\bigwedge W_{\cre}^{(+);N}$
is proportional to $\mathbf{e}_{(+)}^{\top}V^{N}$.
\end{proof}

\subsection{The physical polarization}

The relevant polarization and basis is the one in which the particle-states
$a_{\beta_{1}}^{\dagger}a_{\beta_{2}}^{\dagger}\cdots a_{\beta_{n}}^{\dagger}.v_{{\rm vac}}$
are eigenvectors of the evolution defined by the transfer matrix.
We make use of the fact that $1$ is not an eigenvalue of $T_{V}$,
which follows from Proposition \ref{prop:form-of-the-eigenvalues}
and Theorem \ref{thm:massive-s-hol-prop-induced-rotation}.
\begin{lem}
\label{lem:physical-polarization} Let $\sW_{\cre}^{{\rm phys}}\subset\sW$
be the subspace spanned by eigenvectors of $T_{V}$ with eigenvalues
less than one and $\sW_{\ann}^{{\rm phys}}\subset\sW$ the subspace
spanned by eigenvectors of $T_{V}$ with eigenvalues greater than
one. Then $\mathcal{W}_{\cre}^{{\rm phys}}\oplus\mathcal{W}_{\ann}^{{\rm phys}}$
is a polarization. As a representation of $\Cliff$, the space $\sS_{+}$
is isomorphic to the Fock representation $\bigwedge\sW_{\cre}^{{\rm phys}}$.\end{lem}
\begin{proof}
Recall that for any $u,v\in\sW$ we have $(T_{V}u,T_{V}v)=(u,v)$.
For eigenvectors of $u,v$ of $T_{V}$ it follows that $(u,v)$ can
be non-zero only if the eigenvalues are inverses of each other, and
thus the bilinear form vanishes when restricted to $\sW_{\cre}^{{\rm phys}}$
or $\sW_{\ann}^{{\rm phys}}$. Finally, $\sW=\sW_{\cre}^{{\rm phys}}\oplus\sW_{\ann}^{{\rm phys}}$
because $T_{V}$ is diagonalizable with real eigenvalues and $1$
is not an eigenvalue.
\end{proof}
Then let $(a_{\alpha})_{\alpha=1}^{|\mathbf{I}^{*}|}$ be a basis
of $\sW_{\ann}^{{\rm phys}}$ consisting of eigenvectors of the induced
rotation $T_{V}(a_{\alpha})=\lambda_{\alpha}a_{\alpha}$ with $\lambda_{\alpha}>1$,
and let $ $$(a_{\alpha}^{\dagger})_{\alpha=1}^{|\mathbf{I}^{*}|}$
be the dual basis of $\sW_{\cre}^{{\rm phys}}$, i.e. $(a_{\alpha}^{\dagger},a_{\beta})=\delta_{\alpha,\beta}$.
Note that we have $T_{V}(a_{\alpha}^{\dagger})=\lambda_{\alpha}^{-1}a_{\alpha}^{\dagger}$.
\begin{prop}
\label{prop:transfer-matrix-on-physical-Fock-space} If $v\in\sS$
is an eigenvector of $V$ with eigenvalue $\Lambda$, then the vector
$a_{\alpha}^{\dagger}v\in\sS$ is either zero or an eigenvector with
eigenvalue $\lambda_{\alpha}^{-1}\Lambda$ and $a_{\alpha}v\in\sS$
is either zero or an eigenvector of eigenvalue $\lambda_{\alpha}\Lambda$.
In particular, if $\Lambda_{0}$ is the largest eigenvalue of $V$
and $v_{{\rm vac}}^{{\rm phys}}\in\sS_{+}$ is the corresponding eigenvector,
then $v_{{\rm vac}}^{{\rm phys}}$ is a vacuum of the Fock space $\sS_{+}$
and the vectors $a_{\alpha_{1}}^{\dagger}a_{\alpha_{2}}^{\dagger}\cdots a_{\alpha_{n}}^{\dagger}.v_{{\rm vac}}$
form a basis of $\sS_{+}$ consisting of eigenvectors with eigenvalues
$\Lambda_{0}\times\prod_{s=1}^{n}\lambda_{\alpha_{s}}^{-1}$.\end{prop}
\begin{proof}
For $v\in\sS$ an eigenvector, $Vv=\Lambda v$, compute $Va_{\alpha}v=(Va_{\alpha}V^{-1})Vv=T_{V}(a_{\alpha})Vv=\lambda_{\alpha}\Lambda\, a_{\alpha}v$,
and similarly for $a_{\alpha}^{\dagger}v$. It is then clear that
$v_{{\rm vac}}^{{\rm phys}}$ is annihilated by all of $\sW_{\ann}$,
because $\lambda_{\alpha}\Lambda_{0}$ is larger than the largest
eigenvalue of $V$. \end{proof}
\begin{thm}
\label{thm:transfer-matrix-on-physical-Fock-space} Let $P_{\beta}^{\bC}\colon(\bC^{2})^{\mathbf{I}^{*}}\rightarrow(\bC^{2})^{\mathbf{I}^{*}}$
be the complexified massive s-holomorphic row-to-row propagation,
and let $W_{\circ}\subset(\bC^{2})^{\mathbf{I}^{*}}$ be the subspace
spanned by eigenvectors of $P_{\beta}^{\bC}$ with eigenvalues less
than one. On the exterior algebra $\bigwedge W_{\circ}=\bigoplus_{n=0}^{|\mathbf{I}^{*}|}\wedge^{n}W_{\circ}$
define $\Gamma(P_{\beta}^{\bC})=\bigoplus_{n=0}^{|\mathbf{I}^{*}|}(P_{\beta}^{\bC}|_{W_{\circ}})^{\tens n}$.
Then there is a linear isomorphism $\rho:\sS_{+}\rightarrow\bigwedge W_{\circ}$
such that 
\begin{align*}
\rho\circ V\circ\rho^{-1} & =\;\const\times\Gamma(P_{\mu}^{\bC}).
\end{align*}
\end{thm}
\begin{proof}
The state space $\sS_{+}$ is isomorphic to the Fock space $\bigwedge\sW_{\cre}^{{\rm phys}}$
by Lemma \ref{lem:physical-polarization}. By Proposition \ref{prop:transfer-matrix-on-physical-Fock-space},
in this identification the transfer matrix $V$ becomes diagonal in
the basis $a_{\alpha_{1}}^{\dagger}\wedge\cdots\wedge a_{\alpha_{n}}^{\dagger}$,
with eigenvalues $\Lambda_{0}\prod_{s=1}^{n}\lambda_{\alpha_{s}}^{-1}$,
and thus it coincides with $\bigoplus_{n=0}^{|\mathbf{I}^{*}|}\left(T_{V}|_{W_{\cre}^{{\rm phys}}}\right)^{\tens n}$
apart from the overall multiplicative constant $\Lambda_{0}$. It
remains to note that by Theorem \ref{thm:massive-s-hol-prop-induced-rotation}
the induced rotation $T_{V}\colon\sW\rightarrow\sW$ coincides up
to isomorphism with the complexification $P_{\beta}^{\bC}\colon(\bC^{2})^{\mathbf{I}^{*}}\rightarrow(\bC^{2})^{\mathbf{I}^{*}}$
of the row-to-row propagation, and the same holds for the restrictions
$T|_{\sW_{\cre}^{{\rm phys}}}$ and $P_{\beta}^{\bC}|_{W_{\circ}}$
to the corresponding subspaces.
\end{proof}

\section{\label{sec:op-corr-and-obs}Operator Correlations and Observables}

In this section we discuss correlation functions of operators in the
transfer matrix formalism. We introduce in particular holomorphic
and antiholomorphic fermion operators, and show that they form an
operator valued complexified s-holomorphic function. The low temperature
expansions of the fermion operator correlation functions are simply
expressible in terms of parafermionic observables.

\subsection{Operator insertions in the Ising model transfer matrix formalism}

We consider the Ising model in the rectangle $ $$\mathbf{I}\times\mathbf{J}$,
with $\mathbf{I}=\set{a,a+1,\ldots,b-1,b}$ and $\mathbf{J}=\set{0,1,\ldots,N-1,N}$,
and we denote the row $y$ by $\mathbf{I}_{y}=\mathbf{I}\times\set y$.
We use the notation of Section \ref{sub:symm-tm} for the transfer
matrix (with locally constant boundary conditions on the left and
right sides of the rectangle) and the Clifford algebra.

The total energy of a spin configuration $\mathbf{s}\in\set{\pm1}^{\mathbf{I}\times\mathbf{J}}$
is $H(\mathbf{s})=-\sum_{v\sim w}\mathbf{s}_{v}\mathbf{s}_{w}$, with
the sum over $v,w\in\mathbf{I}\times\mathbf{J}$ that are nearest
neighbors on the square lattice, $|v-w|=1$. The probability measure
of the Ising model with plus boundary conditions is given by $\PR^{+}[\set{\mathbf{s}}]=\frac{1}{\mathcal{Z}^{+}}e^{-\beta H(\mathbf{s})}$
on the set $\set{\mathbf{s}\in\set{\pm1}^{\mathbf{I}\times\mathbf{J}}\,\Big|\,\mathbf{s}|_{\bdry(\mathbf{I}\times\mathbf{J})}\equiv+1}$
of spin configurations that are $+1$ on the boundary of the rectangle.
The normalizing constant in the formula is the partition function
\begin{align*}
\mathcal{Z}^{+}=\; & \sum_{\substack{\mathbf{s}\in\set{\pm1}^{\mathbf{I}\times\mathbf{J}}\\
\mathbf{s}|_{\bdry(\mathbf{I}\times\mathbf{J})}\equiv+1
}
}e^{-\beta H(\mathbf{s})}.
\end{align*}
The partition function can be expressed in terms of the transfer matrix
$V$ by expanding a product of transfer matrices in the basis $(\mathbf{e}_{\sigma})$
indexed by spin configurations in a row, $\sigma\in\set{\pm1}^{\mathbf{I}}$.
More precisely, we have 
\begin{align*}
\mathcal{Z}^{+}=\; & \sum_{\substack{\mathbf{s}\in\set{\pm1}^{\mathbf{I}\times\mathbf{J}}\\
\mathbf{s}|_{\bdry(\mathbf{I}\times\mathbf{J})}\equiv+1
}
}e^{\beta\sum_{v\sim w}\mathbf{s}_{v}\mathbf{s}_{w}}=\mathbf{f}^{\top}V^{N}\mathbf{i}=:\langle\mathbf{f}|V^{N}|\mathbf{i}\rangle,
\end{align*}
where the {}``initial state'' $\mathbf{i}$ and the {}``final state''
$\mathbf{f}$ are given by $\mathbf{i}=\mathbf{f}=\left(V^{{\rm h}}\right)^{\half}\mathbf{e}_{(+)}=e^{\frac{\beta}{2}|\mathbf{I}^{*}|}\mathbf{e}_{(+)}$
--- we included a factor to correctly take into account the interactions
along the horizontal edges in the top and bottom rows.

The spin operators $\hat{\sigma}_{j}\colon\sS\rightarrow\sS$ are
the diagonal matrices in the basis $(\mathbf{e}_{\sigma})$ with diagonal
entries given by the value of $\sigma\in\set{\pm1}^{\mathbf{I}}$
at position $j\in\mathbf{I}$, i.e.
\begin{align*}
\hat{\sigma}_{j}(\mathbf{e}_{\sigma})=\; & \sigma_{j}\,\mathbf{e}_{\sigma}.
\end{align*}

Note that for example the expected value of the spin $\mathbf{s}_{z}$
at $z=x+\ii y\in\mathbf{I}\times\mathbf{J}$, with respect to the
probability measure $\PR^{+}$ of the Ising model with plus boundary
conditions, can be written as
\begin{align*}
\EX^{+}[\mathbf{s}_{z}]=\; & \frac{\sum_{\mathbf{s}}\mathbf{s}_{z}\exp(\beta\sum_{v\sim w}\mathbf{s}_{v}\mathbf{s}_{w})}{\sum_{\mathbf{s}}\exp(\beta\sum_{v\sim w}\mathbf{s}_{v}\mathbf{s}_{w})}=\frac{\langle\mathbf{f}|V^{N-y}\hat{\sigma}_{x}V^{y}|\mathbf{i}\rangle}{\langle\mathbf{f}|V^{N}|\mathbf{i}\rangle},
\end{align*}
by expanding also matrix products in the numerator in the basis $(\mathbf{e}_{\sigma})$.
Moreover, the initial and final states $\mathbf{i},\mathbf{f}\propto\mathbf{e}_{(+)}$
could be replaced by $\mathbf{e}_{(+)}$ because the constants would
cancel in the ratio. Finally, the formula takes a yet simpler form
if we define the time-dependent spin operator
\begin{align*}
\hat{\sigma}(x+\ii y)=\; & V^{-y}\hat{\sigma}_{x}V^{y}
\end{align*}
and indeed it is simple to check that then
\begin{align*}
\EX^{+}\left[\prod_{i=1}^{r}\mathbf{s}_{z_{i}}\right]=\; & \frac{\langle\mathbf{e}_{(+)}|V^{N}\hat{\sigma}(z_{1})\cdots\hat{\sigma}(z_{r})|\mathbf{e}_{(+)}\rangle}{\langle\mathbf{e}_{(+)}|V^{N}|\mathbf{e}_{(+)}\rangle}.
\end{align*}

We define, as in the proof of Theorem \ref{thm:massive-s-hol-prop-induced-rotation},
the Clifford algebra elements $\psi_{k}=\frac{\ii}{\sqrt{2}}(p_{k}+q_{k})\in\sW$,
$\bar{\psi}_{k}=\frac{1}{\sqrt{2}}(p_{k}-q_{k})\in\sW$ for $k\in\mathbf{I}^{*}$.
The corresponding time-dependent operators
\begin{align}
\begin{array}{c}
\psi(k+\ii y)=V^{-y}\psi_{k}V^{y}\\
\bar{\psi}(k+\ii y)=V^{-y}\bar{\psi}_{k}V^{y}
\end{array},\qquad & k\in\mathbf{I}^{*},\; y\in\mathbf{J},\label{eq:def-of-time-dep-fermion-operators}
\end{align}
are called the holomorphic fermion and the anti-holomorphic fermion,
respectively. The reason for this terminology is Theorem \ref{thm:local-relations-for-the-fermion-operators}
below, which states that the pair of operator valued functions $(\psi,\bar{\psi})$
satisfies local linear relations that have the same coefficients as
the defining relations of s-holomorphicity for a function and its
complex conjugate.

The following abbreviated notation 
\begin{align*}
\left\langle \psi^{(1)}(z_{1})\cdots\psi^{(n)}(z_{n})\right\rangle _{\mathbf{I}\times\mathbf{J}}^{+}\,:=\; & \frac{\langle\mathbf{e}_{(+)}|V^{N}\psi^{(1)}(z_{1})\cdots\psi^{(n)}(z_{n})|\mathbf{e}_{(+)}\rangle}{\langle\mathbf{e}_{(+)}|V^{N}|\mathbf{e}_{(+)}\rangle}
\end{align*}
will be used for the correlation functions of the fermion operators,
where $z_{1},\ldots,z_{n}$ are edges, and for each $i=1,2,\ldots n$
we let $\psi^{(i)}$ stand for either $\psi$ or $\bar{\psi}$.

Note that if $R\colon\sW\rightarrow\sW$ is the linear isomorphism
introduced in Section \ref{sub:induced-rotation}, then we have 
\begin{align*}
R(\psi(z))=\; & \bar{\psi}(r(z)), & R(\bar{\psi}(z))=\; & \psi(r(z)),
\end{align*}
where $r(x+\ii y)=a+b-x+\ii y$, and if $J\colon\sW\rightarrow\sW$
is the conjugate-linear isomorphism of the same section, then
\begin{align*}
J(\psi(z))=\; & \bar{\psi}(z), & J(\bar{\psi}(z))=\; & \psi(z).
\end{align*}

\subsection{\label{sub:s-hol-fermion-op}S-holomorphicity of fermion operator}

The fermion operators $\psi(z)$ and $\bar{\psi}(z)$ were defined
in the previous section for $z\in\mathbf{I}^{*}\times\mathbf{J}$,
i.e. on the set of horizontal edges of the rectangle $\mathbf{I}\times\mathbf{J}$.
The following theorem says that we can extend to vertical edges so
that the pair $(\psi,\bar{\psi})$ is a complexified operator valued
(massive) s-holomorphic function.
\begin{thm}
\label{thm:local-relations-for-the-fermion-operators} Let the fermion
operators $\psi(z),\bar{\psi}(z)$ be defined by Equation (\ref{eq:def-of-time-dep-fermion-operators})
for horizontal edges $z\in\mathbf{I}^{*}\times\mathbf{J}$. Then there
exists a unique extension of $\psi$ and $\bar{\psi}$ to the set
of vertical edges $\mathbf{I}\times\mathbf{J}^{*}$, such that the
following local relations hold. For any face, with $E,N,W,S$ the
four edges around the face as in Figure \ref{fig:The-four-edges},
we have

\begin{align}
\psi(N)+\nu^{-1}\lambda\bar{\psi}(N)=\; & \nu^{-1}\psi(E)+\lambda\bar{\psi}(E)\label{eq:fermion-holomorphicity}\\
\psi(N)+\nu\lambda^{-1}\bar{\psi}(N)=\; & \nu\psi(W)+\lambda^{-1}\bar{\psi}(W)\nonumber \\
\psi(S)+\nu\lambda^{3}\bar{\psi}(S)=\; & \nu\psi(E)+\lambda^{3}\bar{\psi}(E)\nonumber \\
\psi(S)+\nu^{-1}\lambda^{-3}\bar{\psi}(S)=\; & \nu^{-1}\psi(W)+\lambda^{-3}\bar{\psi}(W),\nonumber 
\end{align}
and on the left and right boundaries we have
\begin{align}
\psi(a+\ii y)+\ii\,\bar{\psi}(a+\ii y)=\; & 0\label{eq:fermion-boundary-condition}\\
\psi(b+\ii y)-\ii\,\bar{\psi}(b+\ii y)=\; & 0\qquad\text{for any }y\in\mathbf{J}^{*}.\nonumber 
\end{align}
\end{thm}
\begin{rem*}
The coefficients of the linear relations among the operators on incident
edges, Equations (\ref{eq:fermion-holomorphicity}), coincide with
the coefficients in the definition of massive s-holomorphicity, Definition
\ref{def:massive-s-hol}. Similarly, coefficients in the Equations
(\ref{eq:fermion-boundary-condition}) coincide with the equations
defining the Riemann boundary condition on the left and right boundaries.
The situation at the top and bottom boundaries is slightly different:
the operators $\psi$ and $\bar{\psi}$ are linearly independent,
but when the operators are applied to specific boundary states we
recover similar relations, e.g. at the bottom for $z\in\mathbf{I}_{0}$
we have $\left(\psi(z)+\bar{\psi}(z)\right)\mathbf{e}_{(+)}=0$.\end{rem*}
\begin{proof}
The uniqueness of such extension is clear by the following explicit
construction similar to the one in the proof of Lemma \ref{lem:massive-shol-prop}.
Consider the vertical position $y\in\mathbf{J}^{*}$. For $z\in(\mathbf{I}_{y}\setminus\bdry\mathbf{I}_{y})$
one can solve for $\psi(z)$ from Equations (\ref{eq:fermion-holomorphicity})
(the third and fourth equations on the plaquettes on the left and
right of $z$) in terms of the operators $\psi(w)$ and $\bar{\psi}(w)$,
$w\in\mathbf{I}_{y-\half}^{*}$, more precisely in terms of $\psi(z-\half-\frac{\ii}{2}),\,\bar{\psi}(z-\half-\frac{\ii}{2}),\,\psi(z+\half-\frac{\ii}{2}),\,\bar{\psi}(z+\half-\frac{\ii}{2})$.
Similarly one can solve for $\bar{\psi}(z)$ and the result is $J(\psi(z))$.
For $z\in\bdry\mathbf{I}_{y}=\set{a+\ii y,\, b+\ii y}$ on the boundary,
using both Equations (\ref{eq:fermion-holomorphicity}) and (\ref{eq:fermion-boundary-condition})
one can solve for $\psi(z)$ in terms of the operators $\psi(w)$
and $\bar{\psi}(w)$, $w\in\mathbf{I}_{y-\half}^{*}$, more precisely
in terms of $\psi(a+\half-\frac{\ii}{2})$ and $\bar{\psi}(a+\half-\frac{\ii}{2})$
or $\psi(b-\half-\frac{\ii}{2})$ and $\bar{\psi}(b-\half-\frac{\ii}{2})$.
Again similarly $\bar{\psi}(z)=J(\psi(z))$.

Extending $\psi$ and $\bar{\psi}$ to $\mathbf{I}_{y}$ with the
above formulas in terms of $\psi$ and $\bar{\psi}$ in the row $\mathbf{I}_{y-\half}^{*}$,
the Equations (\ref{eq:fermion-boundary-condition}) as well as the
third and fourth of Equations (\ref{eq:fermion-holomorphicity}) hold
by definition. Then note that by a similar argument, there are unique
values of $\psi$ and $\bar{\psi}$ in the row $\mathbf{I}_{y+\half}^{*}$
such that the first and second of Equations (\ref{eq:fermion-holomorphicity})
hold. Since the coefficients of the equations we have used are the
same as the coefficients defining massive s-holomorphicity, the unique
definitions of $\psi$ in the row $\mathbf{I}_{y+\half}^{*}$ are
expressible as linear combinations of the operators in the row $\mathbf{I}_{y-\half}^{*}$
with the same coefficients as in the massive s-holomorphic row-to-row
propagation $P_{\beta}$, in Lemma \ref{lem:s-hol-propagator-formula}.
But by Theorem \ref{thm:massive-s-hol-prop-induced-rotation}, these
linear combinations are just the inverse induced rotations applied
to $\psi$ in the row $\mathbf{I}_{y-\half}^{*}$, i.e. the definitions
of the fermions $\psi$ on horizontal edges in the row $\mathbf{I}_{y+\half}^{*}$.
Again $\bar{\psi}$ is recovered by the application of $J$. This
proves the existence of the extension satisfying the local relations
(\ref{eq:fermion-holomorphicity}) and (\ref{eq:fermion-boundary-condition}).
\end{proof}

\subsection{\label{sub:winding-observables-and-low-T-expansions}Ising parafermionic
observables and low temperature expansions}

\subsubsection{The two-point Ising parafermionic observables\label{subsub:Two-point-observables}}

We next consider graphical expansions of correlation functions of
the fermion operators $\psi(z)$, $\bar{\psi}(z)$. These are expansions
in powers of the parameter $\alpha=e^{-2\beta}$, and they are called
\emph{low temperature expansions} because the parameter is small when
the inverse temperature is large ($\alpha\searrow0$ as $\beta\nearrow\infty$).

Let $a\in\mathbf{I}^{*}\times\mathbf{J}$ be a horizontal edge and
$z\in(\mathbf{I}^{*}\times\mathbf{J})\cup(\mathbf{I}\times\mathbf{J}^{*})$
any edge of the rectangle $\mathbf{I}\times\mathbf{J}$. The set of
faces $\mathbf{I}^{*}\times\mathbf{J}^{*}$ of the rectangle forms
the dual graph, and we denote by $\mathcal{E}^{*}=\set{<p,p'>\,\big|\, p,p'\in\mathbf{I}^{*}\times\mathbf{J}^{*},\;|p-p'|=1}$
the set of dual edges. The low temperature expansions of fermion correlation
functions will be simply expressible in terms of the following two
\emph{parafermionic observables}:
\begin{align*}
f_{a}^{\uparrow}(z)=\; & \frac{1}{\mathcal{Z}}\sum_{\gamma\in\contourset_{a}^{\uparrow}(z)}\alpha^{L(\gamma)}e^{-\frac{\ii}{2}\mathbf{W}(\gamma:a\rightarrow z)}\\
f_{a}^{\downarrow}(z)=\; & \frac{1}{\mathcal{Z}}\sum_{\gamma\in\contourset_{a}^{\downarrow}(z)}\alpha^{L(\gamma)}e^{-\frac{\ii}{2}(\mathbf{W}(\gamma:a\rightarrow z)+\pi)},
\end{align*}
where the notation is as follows:
\begin{itemize}
\item $\contourset_{a}^{\uparrow}(z)$ is the set of collections $\gamma\subset\mathcal{E}^{*}$
of dual edges such that the number of edges of $\gamma$ adjacent
to any face $p\in(\mathbf{I}^{*}\times\mathbf{J}^{*})\setminus\set{a+\ii\half,\, p_{z}^{(\gamma)}}$
is even, and the number of edges adjacent to $a+\ii\half$ and $p_{z}^{(\gamma)}$
is odd, where $p_{z}^{(\gamma)}$ is one of the faces next to $z$.
The set $\contourset_{a}^{\downarrow}(z)$ is defined similarly, but
the exceptional odd parities are now at $a-\ii\half$ and at $p_{z}^{(\gamma)}$
one of the faces next to $z$. We visualize $\gamma$ as in Figure
\ref{fig:parafermion} as a set of loops on the dual graph, together
with a path from $a$ to $z$ starting upwards/downwards from $a$,
by including two {}``half-edges'': from $a$ to $a\pm\ii\half$
and from $p_{z}^{(\gamma)}$ to $z$.
\item For $\gamma\in\contourset_{a}^{\uparrow/\downarrow}(z)$ we let $L(\gamma)=|\gamma|+1$
denote the total length of the loops and the path, where $|\gamma|$
is the cardinality of $\gamma\subset\mathcal{E}^{*}$ and the additional
one is included to account for the the two half-edges.
\item The number $\mathbf{W}(\gamma:a\rightarrow z)$ is the cumulative
angle of turns along a path in $\gamma$ from $a$ to $z$. The path
is not necessarily unique, but if it is chosen in such a way that
no edge is used twice and no self-crossings are made, then one can
show that the winding is well defined modulo $4\pi$ and thus the
factor $e^{-\frac{\ii}{2}\mathbf{W}(\gamma:a\rightarrow z)}$ is well
defined \cite{hongler-smirnov-ii}.
\item $\mathcal{Z}$ is given by $\mathcal{Z}=\sum_{\omega\in\contourset}\alpha^{|\omega|}$,
where $\contourset$ is the set of collections $\omega\subset\mathcal{E}^{*}$
of dual edges such that the number of edges of $\omega$ adjacent
to any face $p\in\mathbf{I}^{*}\times\mathbf{J}^{*}$ is even. We
visualize $\omega$ as a collection of loops. The expression for $\mathcal{Z}$
is the low-temperature expansion of the partition function, and it
is easy to see that $\mathcal{Z}=\mathcal{Z}^{+}\times\const$, where
the constant is $e^{\beta\times|(\mathbf{I}^{*}\times\mathbf{J})\cup(\mathbf{I}\times\mathbf{J}^{*})|}$. 
\end{itemize}
\begin{figure}
\includegraphics[width=8cm]{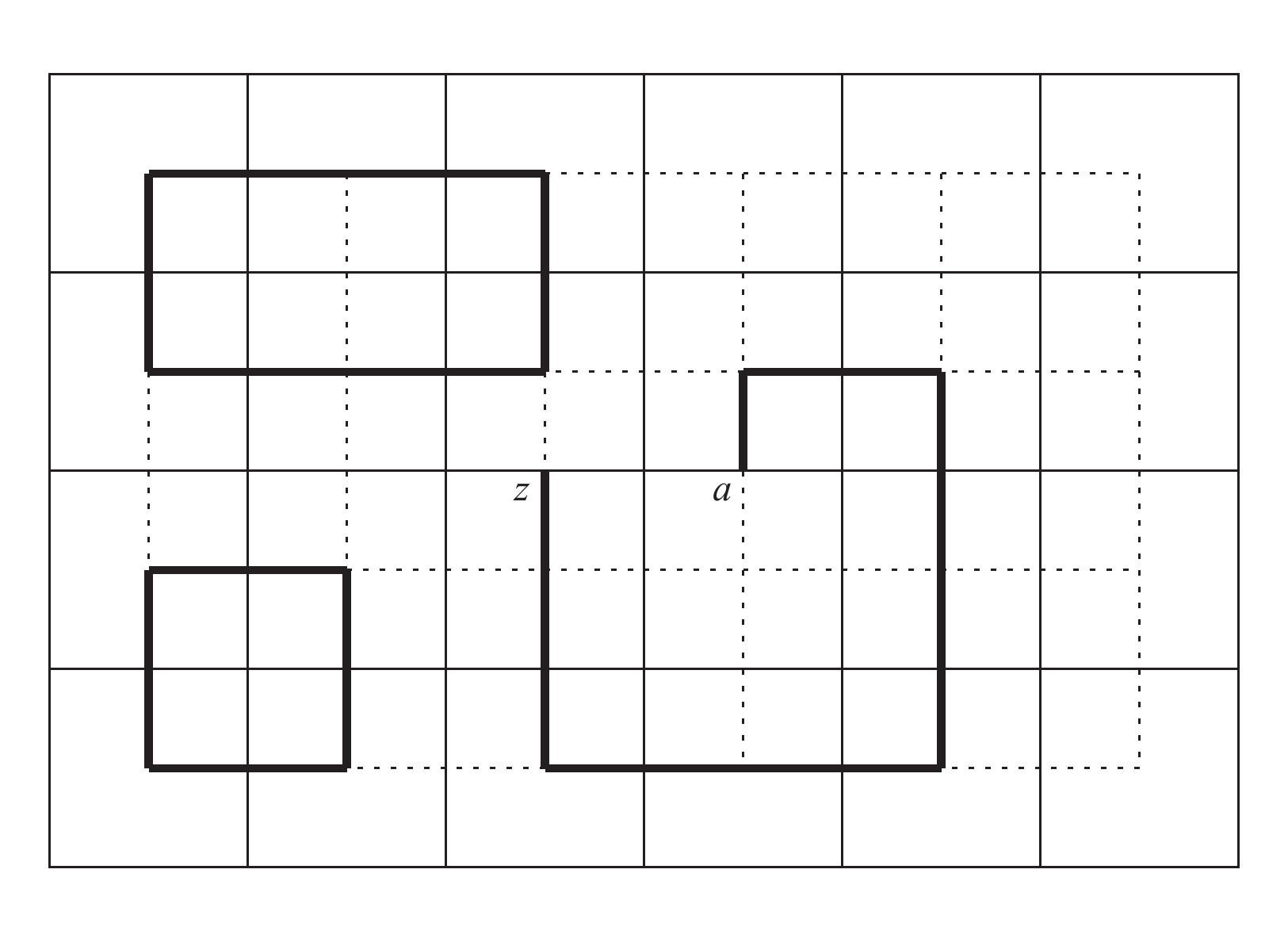}

\caption{\label{fig:parafermion}A configuration in $\contourset_{a}^{\uparrow}(z)$.
The winding in the picture is $\mathbf{W}(\gamma:a\rightarrow z)=-2\pi$.}

\end{figure}

The Ising parafermionic observables are s-holomorphic and satisfy
the Riemann boundary conditions, with a discrete singularity at $z=a$.
To give a more precise statement, we first define a notion of discrete
residue.
\begin{defn}
Let $a$ be a horizontal edge. For a function $z\mapsto f(z)$ that
is (massive) s-holomorphic for $z\neq a$ in a domain containing the
faces $a\pm\frac{\ii}{2}$, the \emph{discrete residue of $f$ at
$a$} is ${\rm Res}_{a}(f)=\frac{\ii}{2\pi}\left(f^{{\rm front}}(a)-f^{{\rm back}}(a)\right)$,
where $f^{{\rm front}}(a)$ is such that if $f$ is extended to $a$
by this value, then $f$ becomes (massive) s-holomorphic on the face
$a+\frac{\ii}{2}$, and $f^{{\rm back}}(a)$ is such that if $f$
is extended to $a$ by this value, then $f$ becomes (massive) s-holomorphic
on the face $a-\frac{\ii}{2}$.\end{defn}
\begin{prop}[\cite{hongler-i}]
 Let $a\in\mathbf{I}^{*}\times\mathbf{J}$. If $a$ is not on the
boundary, $a\in\mathbf{I}^{*}\times(\mathbf{J}\setminus\bdry\mathbf{J})$,
then the Ising parafermionic observables $f_{a}^{\uparrow}$ and $f_{a}^{\downarrow}$
are functions defined on edges $z\neq a$ such that
\begin{itemize}
\item $z\mapsto f_{a}^{\uparrow}(z)$ and $z\mapsto f_{a}^{\downarrow}(z)$
are massive s-holomorphic
\item $f_{a}^{\uparrow}$ and $f_{a}^{\downarrow}$ satisfy RBVP: for $z$
a boundary edge of the rectangle $f_{a}^{\uparrow}(z)\in\bR\tau_{{\rm cw}}^{-\half}$
and $f_{a}^{\downarrow}(z)\in\bR\tau_{{\rm cw}}^{-\half}$
\item the discrete residue of $f_{a}^{\uparrow}$ at $a$ is $\frac{\ii}{2\pi}$
and the discrete residue of $f_{a}^{\downarrow}$ at $a$ is $\frac{-1}{2\pi}$.
\end{itemize}
If $a$ is on the bottom boundary, $a\in\mathbf{I}_{0}^{*}$, then
$f_{a}^{\downarrow}$ is zero and $f_{a}^{\uparrow}$ is a function
defined on edges $z\neq a$ such that
\begin{itemize}
\item $z\mapsto f_{a}^{\uparrow}(z)$ becomes s-holomorphic in the whole
domain with the definition $f_{a}^{\uparrow}(a)=1$
\item $f_{a}^{\uparrow}$ and satisfies RBVP: for $z$ a boundary edge of
the rectangle $f_{a}^{\uparrow}(z)\in\bR\tau_{{\rm cw}}^{-\half}$.
\end{itemize}
If $a$ is on the top boundary , $a\in\mathbf{I}_{N}^{*}$, similar
statements hold.
\end{prop}
The parafermionic observables can be defined similarly in any square
lattice domain \cite{hongler-smirnov-ii}. At the critical point,
$\beta=\beta_{c}$, one can treat scaling limits as follows. Take
the domains to be subgraphs $\Omega_{\delta}$ of the square lattice
$\delta\bZ^{2}$ with small mesh $\delta$, approximating a given
continuous domain $\Omega$ as $\delta\searrow0$. The analogue of
the above Proposition holds. The convergence of the parafermionic
observables as $\delta\searrow0$ can be controlled \cite{hongler-smirnov-ii,hongler-i}:
the functions $f_{a}^{\uparrow}$ and $f_{a}^{\downarrow}$ divided
by $\delta$ converge uniformly on compact subsets of $\Omega\setminus\set a$
to the unique holomorphic function with Riemann boundary values and
the appropriate residue. By Theorems \ref{thm:fermion-operator-two-point-functions}
and \ref{thm:fermion-operator-multi-point-functions} below, we can
deduce from this also the convergence in the scaling limit of the
renormalized fermion correlation functions.

\subsubsection{Fermion operator two-point correlation functions }
\begin{thm}
\label{thm:fermion-operator-two-point-functions} We have 
\begin{align*}
\left\langle \psi(z)\psi(a)\right\rangle _{\mathbf{I}\times\mathbf{J}}=\; & -f_{a}^{\uparrow}(z)+\ii\, f_{a}^{\downarrow}(z)\\
\left\langle \psi(z)\bar{\psi}(a)\right\rangle _{\mathbf{I}\times\mathbf{J}}=\; & \phantom{-}f_{a}^{\uparrow}(z)+\ii\, f_{a}^{\downarrow}(z)\\
=\; & -\overline{f_{z}^{\uparrow}(a)}-\ii\,\overline{f_{z}^{\downarrow}(a)}\\
\left\langle \bar{\psi}(z)\bar{\psi}(a)\right\rangle _{\mathbf{I}\times\mathbf{J}}=\; & -\overline{f_{a}^{\uparrow}(z)}-\ii\,\overline{f_{a}^{\downarrow}(z)}.
\end{align*}
\end{thm}
\begin{proof}
Note that because of the relations of Theorem \ref{thm:local-relations-for-the-fermion-operators}
and the fact that $f_{a}^{\uparrow/\downarrow}$ are massive s-holomorphic,
it suffices to prove the statements when $z$ is a horizontal edge.
Denote $z=x+\ii y$ and $a=x'+\ii y'$. Suppose for simplicity first
that $y>y'$. Consider the numerator of the second correlation function,
\begin{align*}
 & \langle\mathbf{e}_{(+)}|V^{N-y}\psi_{x}V^{y-y'}\bar{\psi}_{x'}V^{y'}|\mathbf{e}_{(+)}\rangle
\end{align*}
Expand the matrix product in the basis $(\mathbf{e}_{\sigma})$. Note
that for any given $\sigma\in\set{\pm1}^{\mathbf{I}}$ the matrix
elements $\left(\bar{\psi}_{x'}\right)_{\tau\sigma}$ and $\left(\psi_{x}\right)_{\tau\sigma}$
are non-zero only if $\tau$ is obtained from $\sigma$ by flipping
the spins on the left of $x'$ or $x$. The expansion is 
\begin{align*}
 & \langle\mathbf{e}_{(+)}|V^{N-y}\psi_{x}V^{y-y'}\bar{\psi}_{x'}V^{y'}|\mathbf{e}_{(+)}\rangle\\
=\; & \const\times\sum V_{(+),\sigma^{(N-1)}}V_{\sigma^{(N-1)},\sigma^{(N-2)}}V_{\sigma^{(N-2)},\sigma^{(N-3)}}\cdots\\
 & \qquad\qquad\cdots V_{\sigma^{(y+1)},\tau^{(y)}}\left(\psi_{x}\right)_{\tau^{(y)},\sigma^{(y)}}V_{\sigma^{(y)},\sigma^{(y-1)}}\cdots\\
 & \qquad\qquad\cdots V_{\sigma^{(y'+1)},\tau^{(y')}}\left(\bar{\psi}_{x'}\right)_{\tau^{(y')},\sigma^{(y')}}V_{\sigma^{(y')},\sigma^{(y'-1)}}\cdots\\
 & \qquad\qquad\cdots V_{\sigma^{(3)},\sigma^{(2)}}V_{\sigma^{(2)},\sigma^{(1)}}V_{\sigma^{(1)},(+)}.
\end{align*}
where the sum is over indices $\sigma^{(1)},\sigma^{(2)},\ldots,\sigma^{(N-1)}\in\set{\pm1}^{\mathbf{I}}$
and the flipped spin configurations are 
\begin{align*}
\tau_{j}^{(y)}=\; & \begin{cases}
\phantom{-}\sigma_{j}^{(y)}\quad & \text{for }j>x\\
-\sigma_{j}^{(y)}\quad & \text{for }j<x.
\end{cases} & \tau_{j}^{(y')}=\; & \begin{cases}
\phantom{-}\sigma_{j}^{(y')}\quad & \text{for }j>x'\\
-\sigma_{j}^{(y')}\quad & \text{for }j<x'.
\end{cases}
\end{align*}
For the matrix elements of $V$ use the formula
\begin{align*}
V_{\rho\sigma}=\; & e^{\frac{\beta}{2}\sum_{k\in\mathbf{I}^{*}}\rho_{k-\half}\rho_{k+\half}}\times e^{\beta\sum_{j\in\mathbf{I}}\sigma_{j}\rho_{j}}\times e^{\frac{\beta}{2}\sum_{k\in\mathbf{I}^{*}}\sigma_{k-\half}\sigma_{k+\half}}\\
=\; & \const\times\alpha^{\half\,\#\{k\,|\,\rho_{k-\half}\neq\rho_{k+\half}\}}\;\alpha^{\#\{j\,|\,\sigma_{j}\neq\rho_{j}\}}\;\alpha^{\half\,\#\{k\,|\,\sigma_{k-\half}\neq\sigma_{k+\half}\}},
\end{align*}
where $\alpha=e^{-2\beta}$. In most rows, we can combine the factors
from the matrix elements of two $\left(V^{{\rm h}}\right)^{\half}$
to just one factor. The sum essentially amounts to summing over spin
configurations in the entire box, except from the peculiarity that
in rows $y$ and $y'$ we have two configurations related to each
other by flipping the spins on the left of $x$ or $x'$. Thus the
terms in the sum correspond to contours $\gamma\in\contourset_{a}^{\uparrow}(z)\cup\contourset_{a}^{\downarrow}(z)$
by the rule that a dual edge is in $\gamma$ if it separates two spins
of opposite value: in rows $y$ and $y'$ the two flipped configurations
amount for half-edges arriving to the points $z=x+\ii y$ and $a=x'+\ii y'$.
The half edge in row $y$ has two possible directions. The half-edge
is either from $x+\ii y$ to the face $x+\ii(y+\half)$ above (resp.
the face $x+\ii(y-\half)$ below) if $\sigma_{x+\half}^{(y)}=\sigma_{x-\half}^{(y)}$
and $\tau_{x+\half}^{(y)}\neq\tau_{x-\half}^{(y)}$ (resp. $\sigma_{x+\half}^{(y)}\neq\sigma_{x-\half}^{(y)}$
and $\tau_{x+\half}^{(y)}=\tau_{x-\half}^{(y)}$) and in this case
we set $\eta=+1$ (resp. $\eta=-1$). Similarly we set $\eta'=+1$
or $\eta'=-1$ if the half edge in row $y'$ is from $x'+\ii y'$
to the face above or below, respectively, i.e. if $\sigma_{x'+\half}^{(y')}=\sigma_{x'-\half}^{(y')}$
or $\sigma_{x'+\half}^{(y')}\neq\sigma_{x'-\half}^{(y')}$, respectively.
The matrix elements of all $V$ together produce a factor $\alpha^{L(\gamma)}$
times a constant. The matrix element of $\psi_{x}$ produces the complex
factor $\ii(-1)^{\#(\gamma\cap\mathbf{I}_{y}^{>x})}\lambda^{\eta}$,
where $\#(\gamma\cap\mathbf{I}_{y}^{>x})$ is the number of edges
of the contour $\gamma$ on row $y$ on the right of $x$ and $\lambda=e^{\ii\pi/4}$.
Similarly the matrix element of $\bar{\psi}_{x'}$ produces the complex
factor $(-1)^{\#(\gamma\cap\mathbf{I}_{y'}^{>x'})}\lambda^{-\eta'}$.
We now write the result of the expansion in terms of sum over contours,
\begin{align*}
 & \langle\mathbf{e}_{(+)}|V^{N-y}\psi_{x}V^{y-y'}\bar{\psi}_{x'}V^{y'}|\mathbf{e}_{(+)}\rangle\\
=\; & \const\times\ii\sum_{\gamma\in\contourset_{a}^{\uparrow}(z)\cup\contourset_{a}^{\downarrow}(z)}\alpha^{L(\gamma)}(-1)^{\#(\gamma\cap\mathbf{I}_{y}^{>x})+\#(\gamma\cap\mathbf{I}_{y'}^{>x'})}\lambda^{\eta-\eta'}.
\end{align*}
Combinatorial considerations of the topological possibilities for
$ $the curve in $\gamma$ from $a$ to $z$ show that $(-1)^{\#(\gamma\cap\mathbf{I}_{y}^{>x})+\#(\gamma\cap\mathbf{I}_{y'}^{>x'})}\lambda^{\eta-\eta'}=-\ii e^{-\frac{\ii}{2}\mathbf{W}(\gamma:a\rightarrow z)}$,
where $\mathbf{W}(\gamma:a\rightarrow z)$ is the winding of the path
as in the definition of the parafermionic observable (note a difference
to the case $y<y'$: we would have $(-1)^{\#(\gamma\cap\mathbf{I}_{y}^{>x})+\#(\gamma\cap\mathbf{I}_{y'}^{>x'})}\lambda^{\eta-\eta'}=\ii e^{-\frac{\ii}{2}\mathbf{W}(\gamma:a\rightarrow z)}$
instead). Thus we write our final expression for the numerator of
the second correlation function,
\begin{align*}
 & \langle\mathbf{e}_{(+)}|V^{N-y}\psi_{x}V^{y-y'}\bar{\psi}_{x'}V^{y'}|\mathbf{e}_{(+)}\rangle\\
=\; & \const\times\sum_{\gamma\in\contourset_{a}^{\uparrow}(z)\cup\contourset_{a}^{\downarrow}(z)}\alpha^{L(\gamma)}e^{-\frac{\ii}{2}\mathbf{W}(\gamma:a\rightarrow z)}.
\end{align*}
The denominator is $\langle\mathbf{e}_{(+)}|V^{N}|\mathbf{e}_{(+)}\rangle=\const\times\mathcal{Z}$
with the same multiplicative constant (here $\mathcal{Z}$ is as in
Section \ref{subsub:Two-point-observables}), so we get the expression
\begin{align*}
\left\langle \psi(z)\bar{\psi}(a)\right\rangle =\; & f_{a}^{\uparrow}(z)+\ii\, f_{a}^{\downarrow}(z).
\end{align*}
In the case $y<y'$, before we do the expansion of the matrix product,
we must anticommute $\psi(z)$ to the right of $\bar{\psi}(a)$, which
gives an overall sign difference. This is nevertheless cancelled in
the end result by another opposite sign resulting from the combinatorial
considerations of topological possibilities for the curve $\gamma$.

For the first correlation function, a similar consideration gives
when $y>y'$, 
\begin{align*}
 & \langle\mathbf{e}_{(+)}|V^{N-y}\psi_{x}V^{y-y'}\psi_{x'}V^{y'}|\mathbf{e}_{(+)}\rangle\\
=\; & \const\times(-1)\sum_{\gamma\in\contourset_{a}^{\uparrow}(z)\cup\contourset_{a}^{\downarrow}(z)}\alpha^{L(\gamma)}(-1)^{\#(\gamma\cap\mathbf{I}_{y}^{>x})+\#(\gamma\cap\mathbf{I}_{y'}^{>x'})}\lambda^{\eta+\eta'}.
\end{align*}
In this case we have $(-1)^{\#\gamma\cap\mathbf{I}_{y}^{>x}+\#\gamma\cap\mathbf{I}_{y'}^{>x'}}\lambda^{\eta+\eta'}=\eta'e^{-\frac{\ii}{2}\mathbf{W}(\gamma:a\rightarrow z)}$,
leading to
\begin{align*}
\left\langle \psi(z)\psi(a)\right\rangle =\; & -f_{a}^{\uparrow}(z)+\ii\, f_{a}^{\downarrow}(z).
\end{align*}

\end{proof}

\subsection{Pfaffian formulas for multi-point fermion correlation functions\label{sub:Pfaffian formula}}

The multi-point correlation functions of the fermions can be written
in terms of two-point correlation functions. Recall the abbreviated
notation of Section \ref{sec:op-corr-and-obs} for fermion correlation
functions --- in particular each $\psi^{(i)}$ in the statement below
can be either $\psi$ or $\bar{\psi}$.
\begin{thm}
\label{thm:fermion-operator-multi-point-functions} We have
\begin{align*}
\left\langle \psi^{(1)}(z_{1})\cdots\psi^{(n)}(z_{n})\right\rangle _{\mathbf{I}\times\mathbf{J}}^{+}=\; & \mathrm{Pf}\left(\left[\left\langle \psi^{(i)}(z_{i})\psi^{(j)}(z_{j})\right\rangle _{\mathbf{I}\times\mathbf{J}}^{+}\right]_{i,j=1}^{n}\right).
\end{align*}
\end{thm}
\begin{proof}
We use the polarization of Lemma \ref{lem:simple-polarization-general-N},
which works for all $\beta>0$ except possibly isolated values, and
since both sides of the asserted equation are analytic as functions
of $\beta$, the statement will be proven for all $\beta$. By the
aforementioned lemma, the state $v_{{\rm vac}}=\mathbf{e}_{(+)}$
is a vacuum of the Fock space $\sS_{+}\isom\bigwedge\sW_{\cre}^{(+);N}$,
and the mapping
\begin{align*}
u\mapsto\; & \frac{1}{\mathbf{e}_{(+)}^{\top}V^{N}\mathbf{e}_{(+)}}\mathbf{e}_{(+)}^{\top}V^{N}u=\left\langle v_{{\rm vac}}^{*},\, u\right\rangle 
\end{align*}
defines the dual vacuum $v_{{\rm vac}}^{*}\in\left(\bigwedge\sW_{\cre}^{(+);N}\right)^{*}$.
The denominator in the definition of correlation functions in Section
\ref{sec:op-corr-and-obs} is the same as the denominator in the above
formula for the dual vacuum, $\langle\mathbf{e}_{(+)}|V^{N}|\mathbf{e}_{(+)}\rangle=\mathbf{e}_{(+)}^{\top}V^{N}\mathbf{e}_{(+)}$.
The correlation functions thus read $\left\langle \psi^{(1)}(z_{1})\cdots\psi^{(n)}(z_{n})\right\rangle _{\mathbf{I}\times\mathbf{J}}^{+}=\left\langle v_{{\rm vac}}^{*},\,\psi^{(1)}(z_{1})\cdots\psi^{(n)}(z_{n})v_{{\rm vac}}\right\rangle $.
Finally note that $\psi^{(i)}(z_{i})\in\sW$ for all $i=1,2,\ldots,n$,
so the statement follows from the fermionic Wick's formula, Lemma
\ref{lem:Wick formula}, applied to the polarization $\sW=\sW_{\cre}^{(+);t}\oplus\sW_{\ann}^{(+);t}$
of Lemma \ref{lem:simple-polarization-general-N}.
\end{proof}

\subsection{\label{sub:Multipoint-parafermionic-observables}Multipoint Ising
parafermionic observables}

Let us now define multipoint parafermionic observables introduced
in \cite{hongler-i}. Let $\Omega$ be a square grid domain, with
dual $\Omega^{*}$ consisting of the faces. Denote the set of edges
of $\Omega$ by $\mathcal{E}$ and the set of dual edges by $\mathcal{E}^{*}$.
Let $z_{1},\ldots,z_{2m}$ be (midpoints of) edges, and for each $z_{j}$,
let $o_{j}$ be a choice of orientation of the corresponding dual
edge $e_{j}^{*}$ (i.e. $o_{j}\in\left\{ \pm1\right\} $ if $e_{j}^{*}$
is horizontal and $o_{j}\in\left\{ \pm i\right\} $ if $e_{j}^{*}$
is vertical), and let $\varepsilon_{j}\in\mathbb{C}$ be choices of
square roots of the orientations, $\varepsilon_{j}^{2}=o_{j}$.

We define the multipoint observable $f^{\epsilon}\left(z_{1},\ldots,z_{2m}\right)$
by
\[
f^{\varepsilon}\left(z_{1},\ldots,z_{2m}\right)=\sum_{\gamma\in\mathcal{C}_{z_{1},\ldots,z_{2m}}^{\varepsilon}}\alpha^{L(\gamma)}\;\mathrm{sign}\left(\gamma\right)\prod_{\pi_{j}:z_{s_{j}}\leadsto z_{d_{j}}}\frac{\varepsilon_{d_{j}}}{\varepsilon_{s_{j}}}e^{-\ii\mathbf{W}\left(\pi_{j}\right)},
\]
where
\begin{itemize}
\item $\mathcal{C}_{z_{1},\ldots,z_{2m}}^{\varepsilon}$ is the set of $\gamma\subset\mathcal{E}^{*}$
consisting of the (dual) half edges $<z_{j},z_{j}+\frac{o_{j}}{2}>$
and of (dual) edges of $\mathcal{E}^{*}$ such that each vertex $p\in\Omega^{*}$
belongs to an even number of edges/half edges of $\gamma$: in other
words a configuration $\gamma$ contains loops and $m$ paths $\pi_{1},\ldots,\pi_{m}$
linking pairwise the $z_{j}$'s. By $L(\gamma)$ we mean the number
of edges of $ $$\mathcal{E}^{*}$ in $\gamma$ plus $m$, with the
additional $m$ accounting for the $2m$ half edges.
\item The product is over the $m$ paths $\pi_{1},\ldots,\pi_{m}$, where
each $\pi_{j}$ is oriented from $z_{s_{j}}$ to $z_{d_{j}}$ where
$s_{j}<d_{j}$ (i.e. we orient the paths from smaller to greater indices).
\item $\mathrm{sign}\left(\gamma\right)=\left(-1\right)^{\#\mathrm{crossings}}$,
where $\#\mathrm{crossings}$ is the number of crossings of the pair
partition $\left\{ \left\{ s_{j},d_{j}\right\} :j\in\left\{ 1,\ldots,m\right\} \right\} $
of $\left\{ 1,\ldots,2m\right\} $ induced by the paths $\pi_{1},\ldots,\pi_{m}$
($\pi_{j}$ from $z_{s_{j}}$ to $z_{d_{j}}$), i.e. the number of
4-tuples $s_{j}<d_{j}<s_{k}<d_{k}$.
\item It can be checked \cite{hongler-i} that if there are ambiguities
in the choices of paths $\pi_{1},\ldots,\pi_{n}$, the weight of a
configuration $\gamma$ is independent of the way that they are resolved,
provided that wherever there is an ambiguity each path turns left
or right (going straight is forbidden).
\end{itemize}
The  observables $f^{\epsilon}\left(z_{1},\ldots,z_{2m}\right)$ can
be used to compute the scaling limit of the energy density correlations,
as well as boundary spin correlations with free boundary conditions
(see \cite{hongler-i}). The key property that allows one to study
the observable at criticality is its s-holomorphicity:
\begin{prop}[\cite{hongler-i}]
\label{prop:multipoint-winding-s-hol}Let $o_{1},\ldots,o_{2m-1}\in\mathbb{C}$
be orientations of edges $e_{1}^{*},\ldots,e_{2m-1}^{*}$ and let
$\varepsilon_{1},\ldots,\varepsilon_{2m-1}\in\mathbb{C}$ be such
that $\varepsilon_{1}^{2}=o_{1},\ldots,\varepsilon_{2m-1}^{2}=o_{2m-1}$.
Let $z_{1},\ldots,z_{2m-1}$ be the midpoints of $e_{1}^{*},\ldots,e_{2m-1}^{*}$.
For any midpoint of edge $z_{2m}$, let $o_{2m}$ and $o_{2m}$ be
its two possible orientations, let $\varepsilon_{2m}$ and $\tilde{\varepsilon}_{2m}$
be such that $\varepsilon_{2m}^{2}=o_{2m}$ and $\tilde{\varepsilon}_{2m}^{2}:=-o_{2m}$,
and let $\varepsilon:=\left(\varepsilon_{1},\ldots,\varepsilon_{2m}\right)$
and $\tilde{\varepsilon}:=\left(\varepsilon_{1},\ldots,\varepsilon_{2m-1},\tilde{\varepsilon}_{2m}\right)$. 

Then we have that $g\left(z_{2m}\right):=\frac{\lambda}{\varepsilon_{2m}}f^{\varepsilon}\left(z_{1},\ldots,z_{2m}\right)+\frac{\lambda}{\tilde{\varepsilon}_{2m}}f^{\tilde{\varepsilon}}\left(z_{1},\ldots,z_{2m}\right)$
is independent of the choice of $\varepsilon_{2m},\tilde{\varepsilon}_{2m}$
and at criticality $z_{2m}\mapsto g\left(z_{2m}\right)$ is s-holomorphic
on $\Omega\setminus\left\{ z_{1},\ldots,z_{2m-1}\right\} $, with
$\parallel\tau_{\mathrm{cw}}^{-\frac{1}{2}}$ boundary conditions.
\end{prop}
As for the fermion operator two point correlation functions and two
point parafermionic observables, it is true that the fermion operator
multipoint correlation functions are expressible as linear combinations
of the multipoint parafermionic observables and vice versa.
\begin{thm}
\label{thm:multipoint-operator-correlations-and-observables}Define
$\psi^{\uparrow}(z)=\frac{1}{2}(\bar{\psi}(z)-\psi(z))$ and $\psi^{\downarrow}(z)=\frac{\ii}{2}(\psi(z)+\bar{\psi}(z))$.
Then we have
\begin{align*}
\left\langle \psi^{\updownarrow_{2m}}(z_{2m})\cdots\psi^{\updownarrow_{1}}(z_{1})\right\rangle _{\mathbf{I}\times\mathbf{J}}^{+}=\; & f^{\eps}(z_{1},\ldots,z_{2m}),
\end{align*}
where the arrows $\updownarrow_{j}\in\set{\uparrow,\downarrow}$ and
the square roots of directions $\eps_{j}=\sqrt{o_{j}}$ are chosen
as follows
\begin{align*}
\begin{cases}
\eps_{j}=\lambda\quad & \text{if }\updownarrow_{j}=\uparrow\\
\eps_{j}=\lambda^{-3}\quad & \text{if }\updownarrow_{j}=\downarrow
\end{cases} & \;.
\end{align*}
\end{thm}
\begin{proof}
Suppose for simplicity that $z_{j}=x_{j}+\ii y_{j}$ with $y_{1}<y_{2}<\cdots<y_{2m}$.
Form a low temperature expansion of the fermion operator correlation
function as in the proof of Theorem \ref{thm:fermion-operator-two-point-functions}.
Consider any fixed $j=1,2,\ldots,2m$. Adding up the low temperature
expansions of the two terms in the definition of $\psi^{\updownarrow_{j}}(z_{j})$
we get that the total weight for the configurations where the half
edge from $z_{j}$ to $z_{j}+\frac{1}{2}\eta_{j}\ii$ is used ($\eta_{j}\in\set{\pm1}$)
is used is in the two cases $\updownarrow_{j}=\uparrow$ and $\updownarrow_{j}=\downarrow$
respectively proportional to $\half(\lambda^{-\eta_{j}}-\ii\lambda^{\eta_{j}})=\lambda^{-1}\delta_{\eta_{j},+1}$
and $\frac{\ii}{2}(\ii\lambda^{\eta_{j}}+\lambda^{-\eta_{j}})=\lambda^{3}\delta_{\eta_{j},-1}$,
where the Kronecker deltas in particular ensure that only the contributions
of the contours $\mathcal{C}_{z_{1},\ldots,z_{2n}}^{\varepsilon}$
survive, as in the definition of the corresponding parafermionic multipoint
observable. In the low temperature expansion, contours $\gamma$ always
come have a factor $\alpha^{L(\gamma)}$ in their weight due to the
product of matrix elements of $V^{{\rm h}}$ and $V^{{\rm v}}$, and
for any surviving contour $\gamma\in\mathcal{C}_{z_{1},\ldots,z_{2n}}^{\varepsilon}$
the remaining phase factor coming from the matrix elements of $\left(\psi_{x_{j}}^{\updownarrow_{j}}\right)$
equals ${\rm sgn}({\rm pairing}(\gamma))\times\prod_{p=1}^{m}\frac{\eps_{d_{p}}}{\eps_{s_{p}}}$. 
\end{proof}
As a direct consequence of Theorems \ref{thm:fermion-operator-multi-point-functions}
and \ref{thm:multipoint-operator-correlations-and-observables}, we
get a Pfaffian formula for the multi-point parafermionic observables. 
\begin{cor}
Let $e_{1}^{*},\ldots,e_{2n}^{*}$ be dual edges with orientations
$o_{1},\ldots,o_{2n}$ and let $\varepsilon_{1},\ldots,\varepsilon_{2n}$
be such that $\varepsilon_{1}^{2}=o_{1},\ldots,\varepsilon_{2n}^{2}=o_{2n}$.
Then we have that
\[
f^{\varepsilon}\left(z_{1},\ldots,z_{2n}\right)=\mathrm{Pfaff}\left(f^{\left(\varepsilon_{j},\varepsilon_{k}\right)}\left(z_{j},z_{k}\right)\mathbf{1}_{j\neq k}\right)_{1\leq j,k\leq2n}.
\]

\end{cor}
This formula was proved in \cite{hongler-i} in the critical case
$\beta=\beta_{c}$ for domains of arbitrary shape, by verifying that
the s-holomorphic function in Proposition \ref{prop:multipoint-winding-s-hol}
satisfies a discrete Riemann boundary value problem with singularities,
which uniquely characterizes the parafermionic observable. The special
case which gives the Ising model boundary spin correlation functions
with free boundary conditions was proven by direct combinatorial methods
in \cite{groeneveld-boel-kasteleyn,kager-lis-meester} for very general
classes of planar graphs. Our approach works at any $\beta$, but
for the Ising model on square lattice only, and for domains of general
shape some minor technical modifications are needed in the proof:
the domain should be thought of as a subgraph of a large rectangle
$\mathbf{I}\times\mathbf{J}$, and for every row the transfer matrix
should be replaced by a composition of $V$ and a projection which
enforces plus boundary conditions outside the domain. Nevertheless,
we believe that our approach in conceptually the clearest, as the
Pfaffian appears simply because of the fermionic Wick's formula (Lemma
\ref{lem:Wick formula}). This illustrates an advantage of the operator
formalism, some algebraic structures underlying the Ising model are
more evident and can be better exploited.

\subsection{Correlation functions of the fermion and spin operators\label{sub:fermion-and-spin-correlations}}

It is also possible to consider correlation functions of fermion operators
and spin operators simultaneously. It turns out that as functions
of the fermion operator positions, these become branches of multivalued
observables. For example, when $a\in\mathbf{I}_{0}^{*}$ is on the
bottom side of the rectangle and $w_{1},\ldots,w_{n}\in\mathbf{I}\times\mathbf{J}$,
\begin{align*}
z\mapsto & \frac{\left\langle \psi(z)\bar{\psi}(a)\hat{\sigma}(w_{1})\cdots\hat{\sigma}(w_{n})\right\rangle _{\mathbf{I}\times\mathbf{J}}^{+}}{\left\langle \hat{\sigma}(w_{1})\cdots\hat{\sigma}(w_{n})\right\rangle _{\mathbf{I}\times\mathbf{J}}^{+}}=\frac{\langle\mathbf{e}_{(+)}|V^{N}\psi(z)\bar{\psi}(a)\hat{\sigma}(w_{1})\cdots\hat{\sigma}(w_{n})|\mathbf{e}_{(+)}\rangle}{\langle\mathbf{e}_{(+)}|V^{N}\hat{\sigma}(w_{1})\cdots\hat{\sigma}(w_{n})|\mathbf{e}_{(+)}\rangle}
\end{align*}
becomes a (massive) s-holomorphic function in the complement of the
branch cuts starting from each $w_{j}\in\mathbf{I}\times\mathbf{J}$
to the right boundary of the rectangle. The function can be extended
to the branch cut in two ways, one of which satisfies (massive) s-holomorphicity
conditions on the faces below the cut and another which satisfies
them on the faces above the cut --- the two definitions differ by
a sign, indicating a square root type monodromy of the function at
the locations $w_{j}$ of the spin insertions. A low temperature expansion
like in the proofs of Theorems \ref{thm:fermion-operator-two-point-functions}
and \ref{thm:multipoint-operator-correlations-and-observables} shows
that this function is a branch of the parafermionic spinor observable
of \cite{chelkak-izyurov,chelkak-hongler-izyurov}, where the observable
is properly defined on a double covering of the punctured lattice
domain in order to obtain a well defined s-holomorphic function.

\section{\label{sec:RPS-Operators}Operators on Cauchy data spaces}

As explained above, the Ising transfer matrix can be constructed directly
in terms of the s-holomorphic propagator (Sections \ref{sub:intro-ising-tm-shol}
and \ref{sub:fock-representations}). We now discuss s-holomorphic
approaches to the data carried by the transfer matrix quantum states.
In Sections \ref{sub:winding-observables-and-low-T-expansions} and
\ref{sub:Multipoint-parafermionic-observables}, we learned the following:
\begin{itemize}
\item The correlation functions of the fermion operators can be expressed
as linear combinations of parafermionic observables. 
\item The parafermionic observables can be characterized in s-holomorphic
terms: they are the unique s-holomorphic functions with Riemann boundary
values and prescribed singularities.
\end{itemize}
In this section, we present an s-holomorphic construction inspired
by transfer matrix states. A quantum state $\mathcal{Q}\in\mathcal{S}$
living on a row $\mathbf{I}_{k}$ contains all the information about
the geometry of the domain and the operator insertions below $\mathbf{I}_{k}$.
Likewise, we construct discrete Riemann Poincaré-Steklov (RPS) operators
living on a row (more generally, any crosscut of the domain), which
act on Cauchy data spaces. These RPS operators together with the vectors
on which they act contain all the information about the geometry of
the domain and operator insertions. These operators can be written
as convolution operators with parafermionic observables, which are
fermion correlations and hence can be directly represented from the
quantum states. They also can be propagated using explicit convolution
operators. 

A great advantage of the discrete RPS operators is that they have
nice scaling limits, as singular integral operators, and that they
work in arbitrary planar geometries.

\subsection{\label{sub:Discrete-RPS-operators}Discrete RPS operators}

In this Section, we define the Riemann Poincaré-Steklov operators.

Let $\Omega$ be a square grid domain, let $\mathfrak{b}\subset\partial\mathcal{E}$
be a collection of boundary edges and $\mathcal{R}_{\Omega}^{\mathfrak{b}}$
and be the space of functions $f:\mathfrak{b}\to\mathbb{C}$ such
that $f\parallel\tau_{\mathrm{ccw}}^{-\frac{1}{2}}$ on $\mathfrak{b}$.
Let $\mathcal{I}_{\Omega}^{\mathfrak{b}}$ be the space of functions
$f:\mathfrak{b}\to\mathbb{C}$ such that $f\parallel\tau_{\mathrm{cw}}^{-\frac{1}{2}}$
on $\mathfrak{b}$.

First, we state a key lemma, which guarantees the uniqueness of solutions
to Riemann boundary value problems.
\begin{lem}
\label{lem:uniqueness-rbvp-sol} Let $\Omega$ be a square grid domain
with edges $\mathcal{E}$. If $h:\mathcal{E}\to\mathbb{C}$ is an
s-holomorphic function with $h\parallel\tau_{\mathrm{cw}}^{-\frac{1}{2}}$
on $\partial\mathcal{E}$, then $h=0$.\end{lem}
\begin{proof}
The proof of this lemma is given in \cite[Corollary 29]{hongler-i}
(where the notion of s-holomorphicity comes with a phase change of
$e^{\ii\pi/4}$ compared to the present paper). The idea is to show
that for any s-holomorphic function $g:\mathcal{E}\to\mathbb{C}$
with boundary values $u+v$, where $u\in\mathcal{R}_{\Omega}^{\mathfrak{b}}$
and $v\in\mathcal{I}_{\Omega}^{\mathfrak{b}}$, we have that $\sum_{z\in\mathfrak{b}}\left|v\left(z\right)\right|^{2}\leq\sum_{z\in\mathfrak{b}}\left|u\left(z\right)\right|^{2}$
(the proof of this inequality relies on the definition of a discrete
analogue of $\Im\mathfrak{m}\int g^{2}$). In our case $u=0$, and
hence $v=0$ as well. \end{proof}
\begin{lem}
\label{lem:rps-existence} For any $u\in\mathcal{R}_{\Omega}^{\mathfrak{b}}$,
there exists a unique $v\in\mathcal{I}_{\Omega}^{\mathfrak{b}}$ such
that $u+v$ has an s-holomorphic extension $h:\mathcal{E}\to\mathbb{C}$
satisfying $h\parallel\tau_{\mathrm{cw}}^{-\frac{1}{2}}$ on $\partial\Omega\setminus\mathfrak{b}$.\end{lem}
\begin{proof}
For any $v\in\mathcal{I}_{\Omega}^{\mathfrak{b}}$, there exists at
most one $u\in\mathcal{R}_{\Omega}^{\mathfrak{b}}$ such that $u+v$
has an s-holomorphic extension $h$ to $\Omega\to\mathbb{C}$ with
$\parallel\tau_{\mathrm{cw}}^{-\frac{1}{2}}$ on $\partial\mathcal{E}\setminus\mathfrak{b}$:
if we suppose there are two extensions, their difference will satisfy
the boundary condition $\parallel\tau_{\mathrm{cw}}^{-\frac{1}{2}}$
on $\partial\mathcal{E}$ and hence be $0$ by Lemma \ref{lem:uniqueness-rbvp-sol}.
By a dimensionality argument, there exists exactly one such $u$ and
the mapping $v\mapsto u$ is an invertible linear map. \end{proof}
\begin{defn}
\label{def:rps-operator}We define the RPS operator $U_{\Omega}^{\mathfrak{b}}:\mathcal{R}_{\Omega}^{\mathfrak{b}}\to\mathcal{I}_{\Omega}^{\mathfrak{b}}$
as the mapping $u\mapsto v$ defined by Lemma \ref{lem:rps-operator-as-conv}
(which is an isomorphism by the proof). \end{defn}
\begin{lem}
\label{lem:rps-operator-as-conv}With the notation of Lemma \ref{lem:rps-existence},
we have that 
\begin{eqnarray}
v\left(x\right) & = & \sum_{y\in\mathfrak{b}\setminus\left\{ x\right\} }u\left(y\right)f_{\Omega}\left(y,x\right)\quad\forall x\in\partial\mathcal{E}\label{eq:conv-form-u-operator}
\end{eqnarray}
 and the s-holomorphic extension $h$ is given by

\begin{equation}
h\left(x\right)=\sum_{y\in\mathfrak{b}}u\left(y\right)f_{\Omega}\left(y,z\right)\quad\forall x\in\mathcal{E}.\label{eq:conv-form-bulk-u-operator}
\end{equation}
\end{lem}
\begin{rem*}
The convolution formula, Equation (\ref{eq:conv-form-bulk-u-operator}),
is the key to pass to the scaling limit: the kernel $f_{\Omega}$
converges to a kernel of a continuous singular integral operator (see
\cite[Section 13]{hongler-kytola}).\end{rem*}
\begin{proof}
Let us first notice that (\ref{eq:conv-form-bulk-u-operator}) implies
(\ref{eq:conv-form-u-operator}): if $x\in\partial\mathcal{E}$, we
have $u\left(x\right)f_{\Omega}\left(x,x\right)\in\tau_{\mathrm{ccw}}^{-\frac{1}{2}}$
and $u\left(y\right)f_{\Omega}\left(y,x\right)\in\tau_{\mathrm{cw}}^{-\frac{1}{2}}\left(y\right)$
for $y\in\partial\mathcal{E}\setminus\left\{ x\right\} $, and hence
the projection of (\ref{eq:conv-form-bulk-u-operator}) on $\mathcal{I}$
is indeed (\ref{eq:conv-form-u-operator}). 

To prove (\ref{eq:conv-form-bulk-u-operator}), notice that the right-hand
side is an s-holomorphic function, with the right boundary conditions
(i.e. the same as $h$ in Lemma \ref{lem:rps-existence}). The difference
of both sides must then be $0$, by Lemma \ref{lem:uniqueness-rbvp-sol}.
\end{proof}
In rectangular boxes, the RPS operator can be written simply in terms
of the s-holomorphic propagation
\begin{lem}
\label{lem:ps-with-tm}Let $\Omega$ be a rectangular box $\mathbf{I}\times\left\{ 0,\ldots,N\right\} $,
let $\mathfrak{b}=\mathbf{I}\times\set 0$ be the bottom side and
let $P:\left(\mathbb{R}^{2}\right)^{\mathbf{I}^{*}}\to\left(\mathbb{R}^{2}\right)^{\mathbf{I}^{*}}$
be the s-holomorphic propagation as defined in Section \ref{sub:s-hol-prop}.
For $N\geq0$, decompose the s-holomorphic propagation $P^{N}$ into
four $\left|\mathbf{I}^{*}\right|\times\left|\mathbf{I}^{*}\right|$
blocks 
\[
P^{N}=\begin{pmatrix}P_{\Re\Re}^{N} & P_{\Re\Im}^{N}\\
P_{\Im\Re}^{N} & P_{\Im\Im}^{N}
\end{pmatrix}
\]
corresponding to the decomposition $\left(\mathbb{R}^{2}\right)^{\mathbf{I}^{*}}\isom\mathbb{R}^{\mathbf{I}^{*}}\oplus\ii\mathbb{R}^{\mathbf{I}^{*}}$
into real and imaginary parts. Then have 
\begin{eqnarray*}
U_{\Omega}^{\mathfrak{b}} & = & -\left(P_{\Im\Im}^{N}\right)^{-1}P_{\Re\Im}^{N}.
\end{eqnarray*}
\end{lem}
\begin{proof}
Let $u\in\mathcal{R}_{\Omega}^{\mathfrak{b}}$ (i.e. purely real in
this case) and $v\in\mathcal{I}_{\Omega}^{\mathfrak{b}}$ be defined
by $v=U_{\Omega}^{\mathfrak{b}}u$. By definition of $U_{\Omega}^{\mathfrak{b}}$,
we have that 
\[
\begin{pmatrix}P_{\Re\Re}^{N} & P_{\Re\Im}^{N}\\
P_{\Im\Re}^{N} & P_{\Im\Im}^{N}
\end{pmatrix}\begin{pmatrix}u\\
v
\end{pmatrix}=\begin{pmatrix}w\\
0
\end{pmatrix}
\]
for some purely real function $w:\mathbf{I}\times\left\{ n\right\} \to\mathbb{R}$.
Hence, we get that $P_{\Im\Re}^{N}u+P_{\Im\Im}^{N}v=0$. Since we
know that for any $u\in\mathcal{R}_{\Omega}^{\mathfrak{b}}$, there
exists a unique $v$ satisfying this equation (as $U_{\Omega}^{\mathfrak{b}}$
is an isomorphism), we get that $v=-\left(P_{\Im\Im}^{N}\right)^{-1}P_{\Re\Im}^{N}u$. 
\end{proof}

\subsection{\label{sub:rps-pairings}RPS pairings }

In this subsection, we explain how to pair together s-holomorphic
data coming from two adjacent domains with disjoint interiors, using
RPS operators and Ising parafermionic observables. A related discussion
about discrete kernel gluings (in a different framework, without boundaries)
can be found in \cite{dubedat-i}. In our framework, the gluing operation
arises as an analogue of pairing of transfer matrix states. 

Let us define the setup of this subsection. Let $\Omega_{1},\Omega_{2}$
be two adjacent square grid domains with disjoint interiors, with
edges $\mathcal{E}_{1},\mathcal{E}_{2}$, let $\Omega:=\Omega_{1}\cup\Omega_{2}$
and assume that $\mathfrak{b}:=\partial\mathcal{E}_{1}\cap\partial\mathcal{E}_{2}$
is connected. Let $U_{1}:=U_{\Omega_{1}}^{\mathfrak{b}}$ and $U_{2}:=U_{\Omega_{2}}^{\mathfrak{b}}$
be the RPS operators defined in the previous subsection and set $\mathcal{R}_{j}:=\mathcal{R}_{\Omega_{j}}^{\mathfrak{b}}$
and $\mathcal{I}_{j}^{\mathfrak{b}}:=\mathcal{I}_{\Omega_{j}}^{\mathfrak{b}}$
for $j=1,2$. We have $\mathcal{R}_{1}=\mathcal{I}_{2}$ and $\mathcal{R}_{2}=\mathcal{I}_{1}$. 
\begin{lem}
\label{lem:id-minus-u1u2-isom} We have that $\left(\mathrm{Id}-U_{1}U_{2}\right):\mathcal{R}_{2}\to\mathcal{R}_{2}$
and $\left(\mathrm{Id}-U_{2}U_{1}\right):\mathcal{R}_{1}\to\mathcal{R}_{1}$
are isomorphisms.\end{lem}
\begin{proof}
The injectivity (and hence the bijectivity) of these operators follows
from the fact that if a function $u\in\mathcal{R}_{2}$ is a fixed
point of $U_{1}U_{2}$, then $u+U_{2}u$ admits an s-holomorphic extension
to $\Omega$ with boundary condition $\parallel\tau_{\mathrm{cw}}^{-\frac{1}{2}}$
on $\partial\mathcal{E}$, and is hence $0$ by Lemma \ref{lem:uniqueness-rbvp-sol}.
\end{proof}
A useful corollary of the previous lemma is the following fixed point
result:
\begin{cor}
\label{cor:affine-ops-unique-fixed-points}Let $h_{1}\in\mathcal{I}_{1}$
and $h_{2}\in\mathcal{I}_{2}$. Then there exists a unique function
$f:\mathcal{E}\to\mathbb{C}$ such that for $j=1,2$, the function
$f-h_{j}$ has an s-holomorphic extension to $\mathcal{E}_{j}$ with
boundary conditions $\parallel\tau_{\mathrm{cw}}^{-\frac{1}{2}}$
on $\partial\mathcal{E}_{j}$. We have that $f=u_{1}+u_{2}$, where
$u_{1}$ and $u_{2}$ are given by
\begin{eqnarray*}
u_{1} & = & \left(\mathrm{Id}-U_{2}U_{1}\right)^{-1}\left(U_{2}h_{1}+h_{2}\right),\\
u_{2} & = & \left(\mathrm{Id}-U_{1}U_{2}\right)^{-1}\left(U_{1}h_{2}+h_{1}\right).
\end{eqnarray*}
\end{cor}
\begin{proof}
Suppose first that there exists an $f$ such that the functions $f-h_{j}$
have s-holomorphic extensions to $\mathcal{E}_{j}$ with boundary
conditions $\parallel\tau_{\mathrm{cw}}^{-\frac{1}{2}}$ on $\partial\mathcal{E}_{j}$.
Set $f_{j}:=f-h_{j}$ and write $f_{j}=u_{j}+v_{j}$, where $u_{j}\in\mathcal{R}_{j}$
and $v_{j}\in\mathcal{I}_{j}$. We have that $f=u_{1}+u_{2}$ and
\begin{eqnarray*}
u_{1} & = & v_{2}+h_{2}=U_{2}u_{2}+h_{2}=U_{2}\left(v_{1}+h_{1}\right)+h_{2}=U_{2}U_{1}u_{1}+U_{2}h_{1}+h_{2}\\
u_{2} & = & v_{1}+h_{1}=U_{1}u_{1}+h_{1}=U_{1}\left(v_{2}+h_{2}\right)+h_{1}=U_{1}U_{2}u_{2}+U_{1}h_{2}+h_{1},
\end{eqnarray*}
which gives that $\left(\mathrm{Id}-U_{2}U_{1}\right)u_{1}=U_{2}h_{1}+h_{2}$
and $\left(\mathrm{Id}-U_{1}U_{2}\right)u_{2}=U_{1}h_{2}+h_{1}$.
By Lemma \ref{lem:id-minus-u1u2-isom}, we obtain the asserted formulas
for $u_{1}$ and $u_{2}$, proving the uniqueness of $f$. For any
$(h_{1},h_{2})$ we have seen that there is at most one and hence
exactly one solution $(u_{1},u_{2})$ of the equations
\begin{align*}
\left[\begin{array}{cc}
U_{1} & -\id\\
-\id & U_{2}
\end{array}\right]\left[\begin{array}{c}
u_{1}\\
u_{2}
\end{array}\right]=\; & \left[\begin{array}{c}
-h_{1}\\
-h_{2}
\end{array}\right],
\end{align*}
showing the existence of $f$ with the desired properties.
\end{proof}
As illustrated in the next subsection, Corollary \ref{cor:affine-ops-unique-fixed-points}
has the following consequence: the value (on $\mathfrak{b}$) of an
s-holomorphic observable $h:\Omega\to\mathbb{C}$ with Riemann boundary
conditions and prescribed singularities in $\Omega_{1},\Omega_{2}$
can be recovered from $\left(h_{1},U_{1}\right)$ and $\left(h_{2},U_{2}\right)$.
In other words, these pairs carry all the relevant information about
$\Omega_{1}$, $\Omega_{2}$ that is needed to compute s-holomorphic
correlations. More precisely, and as will be illustrated in the next
subsection: $U_{1},U_{2}$ encode the geometry of the domain and $h_{1},h_{2}$
encode the singularities (in practice, they are the restriction to
$\mathfrak{b}$ of functions with singularities in $\Omega_{1},\Omega_{2}$
and $\tau_{\mathrm{cw}}^{-\frac{1}{2}}$ boundary conditions on $\partial\mathcal{E}_{1},\partial\mathcal{E}_{2}$). 

Once the values of an s-holomorphic observable $h$ on $\mathfrak{b}$
are known, one can compute the values of $h-h_{1}$ on $\mathcal{E}_{1}$
and the values of $h-h_{2}$ on $\mathcal{E}_{2}$, using the convolution
formula (\ref{eq:conv-form-bulk-u-operator}) in Lemma \ref{lem:rps-operator-as-conv}.

\subsection{\label{sub:fermion-correlations}Fermion correlation and fixed point
problems Cauchy data}

The fermion correlator/parafermionic observable fits naturally in
the framework of the previous subsection. A first consequence is the
following.
\begin{prop}
\label{prop:fermions-on-cut}With the notation of Section \ref{sub:rps-pairings},
set \textup{$Q:=\left(\mathrm{Id}-U_{1}U_{2}\right)^{-1}$. }For any
$x\in\partial\mathcal{E}_{1}\setminus\mathfrak{b}$, we have 
\[
f_{\Omega}\left(x,\cdot\right)\Big|_{\mathfrak{b}}=\left(\mathrm{Id}+U_{2}\right)Q\, f_{\Omega_{1}}\left(x,\cdot\right)\Big|_{\mathfrak{b}}.
\]
\end{prop}
\begin{proof}
Set $f:=f_{\Omega}\left(x,\cdot\right)\Big|_{\mathfrak{b}}$, $f_{1}:=f_{\Omega_{1}}\left(x,\cdot\right)\Big|_{\mathfrak{b}}$.
Write $f=u_{1}+u_{2}$ where $u_{j}\in\mathcal{R}_{j}$. Applying
Corollary \ref{cor:affine-ops-unique-fixed-points} to $f$, $h_{1}:=f_{1}$
and $h_{2}=0$, we get $u_{2}=Qf_{1}$. Since $f_{\Omega}\left(x,\cdot\right)$
is s-holomorphic with $\tau_{\mathrm{cw}}^{-\frac{1}{2}}$ on $\partial\mathcal{E}_{2}\setminus\mathfrak{b}$,
we have that $f=u_{2}+U_{2}u_{2}$.
\end{proof}
Extending $f_{\Omega}\left(x,\cdot\right)$ to $\mathcal{E}_{2}$
gives in particular the following nice formula:
\begin{thm}
\label{thm:pairing-fermions}For any $x\in\partial\mathcal{E}_{1}\setminus\mathfrak{b}$
and $y\in\mathcal{E}_{2}$, we have 
\begin{equation}
f_{\Omega}\left(x,y\right)=f_{\Omega_{2}}\left(\cdot,y\right)\Big|_{\mathfrak{b}}^{\top}\; Q\; f_{\Omega_{1}}\left(x,\cdot\right)\Big|_{\mathfrak{b}}.\label{eq:tm-type-pairing}
\end{equation}
\end{thm}
\begin{proof}
Set $f:=f_{\Omega}\left(x,\cdot\right)\Big|_{\mathfrak{b}}$ and write
$f=u_{2}+v_{2}$, with $u_{2}\in\mathcal{R}_{2}$ and $v_{2}\in\mathcal{I}_{2}$.
By Lemma \ref{lem:rps-operator-as-conv}, we have that $f_{\Omega}\left(x,y\right)=\sum_{z\in\mathfrak{b}}u_{2}\left(z\right)f_{\Omega_{2}}\left(z,y\right)$.
By Proposition \ref{prop:fermions-on-cut}, $u_{2}=Qf_{\Omega_{1}}\left(x,\cdot\right)\Big|_{\mathfrak{b}}$
and the result follows. 
\end{proof}
The formula is an analogue of a natural pairing in transfer matrix
formalism: for example when $\Omega=\mathbf{I}\times\set{0,1,\ldots,N}$,
and $x\in\mathbf{I}_{0}$, $x'+\ii N\in\mathbf{I}_{N}$, the correlation
function $\left\langle \bar{\psi}\left(x\right)\psi\left(x'+\ii N\right)\right\rangle $
can be obtained by propagating the state $\bar{\psi}_{x}\mathbf{i}$
with $V^{k}$ to the $k$:th row and pairing it (with the inner product
in $\mathcal{S}$) with the state $\psi_{x'}^{\top}\mathbf{f}$ propagated
downwards from row $N$ to row $k$: Equation \ref{eq:tm-type-pairing}
is the analogue of 
\begin{align*}
\left\langle \bar{\psi}\left(x\right)\psi\left(x'+\ii N\right)\right\rangle =\; & \frac{\left(\psi_{x'}^{\top}\mathbf{f}\right)^{\top}\, V^{N}\;\bar{\psi}_{x}\mathbf{i}}{\mathbf{f}^{\top}\, V^{N}\;\mathbf{i}}.
\end{align*}

When $x$ and $y$ are both in $\mathcal{E}_{1}$, we can pair the
states associated with $x$ and $y$ (this is not possible with transfer
matrix):
\begin{prop}
\label{prop:pairing-same-domain}For $x\in\partial\mathcal{E}_{1}\setminus\mathfrak{b}$
and $y\in\mathcal{E}_{1}$, we have
\[
f_{\Omega}\left(x,y\right)=f_{\Omega_{1}}\left(x,y\right)+\left(f_{\Omega_{1}}\left(\cdot,y\right)\Big|_{\mathfrak{b}}\right)^{\top}\:\left(\mathrm{Id}-U_{1}U_{2}\right)^{-1}U_{2}\:\left(f_{\Omega_{1}}\left(x,\cdot\right)\Big|_{\mathfrak{b}}\right).
\]
\end{prop}
\begin{proof}
Set $f=f_{\Omega}\left(x,\cdot\right)\Big|_{\mathfrak{b}}$ and write
$f=u_{1}+v_{1}$, with $u_{1}\in\mathcal{R}_{1}$ and $v_{1}\in\mathcal{I}_{1}$.
By Lemma \ref{lem:rps-operator-as-conv}, we have that 
\[
f_{\Omega}\left(x,y\right)-f_{\Omega_{1}}\left(x,y\right)=\sum_{z\in\mathfrak{b}}u_{1}\left(z\right)f_{\Omega_{1}}\left(z,y\right)
\]
As in the proof of Proposition \ref{prop:fermions-on-cut}, with Corollary
\ref{cor:affine-ops-unique-fixed-points} applied to $h_{1}=f_{\Omega_{1}}\left(x,\cdot\right)\Big|_{\mathfrak{b}}$
and $h_{2}=0$ we get that $u_{1}=\left(\mathrm{Id}-U_{1}U_{2}\right)^{-1}U_{2}\:\left(f_{\Omega_{1}}\left(x,\cdot\right)\Big|_{\mathfrak{b}}\right)$
and the result follows.
\end{proof}

\thanks{\textbf{Acknowledgements:} Work supported by NSF grant DMS-1106588
and the Minerva Foundation and Academy of Finland grant {}``Conformally
invariant random geometry and representations of infinite dimensional
Lie algebras''. We thank Dmitry Chelkak, Julien Dubédat, John Palmer,
Duong H. Phong and Stanislav Smirnov for interesting discussions.}

\bigskip{}

\end{document}